\documentclass[11pt]{article}
\usepackage[margin=1in]{geometry}
\usepackage{setspace}
\usepackage{mathrsfs}
\usepackage{float}

\newcommand{\PATH}{} 

\usepackage{\PATH slivkins-setup}
\usepackage{\PATH slivkins-theorems}

\usepackage[round]{natbib}
\usepackage{nicefrac}
\usepackage{kpfonts}
\usepackage{array}
\usepackage{booktabs}
\usepackage{graphicx}
\usepackage[small,bf]{caption}
\usepackage{subcaption}
\usepackage{color}
\usepackage{ifthen}
\usepackage{xspace}
\usepackage[noend]{algorithmic}
\usepackage{algorithm}
\usepackage[breaklinks]{hyperref}
\usepackage{url}
\usepackage{tocbibind}
\usepackage{enumerate}
\usepackage{mdframed}
\usepackage{comment}
\usepackage{tikz}
\usetikzlibrary{decorations.pathreplacing}
\usepackage{cleveref}

\usepackage{wrapfig}

\usepackage[toc,title,titletoc]{appendix}

\newcommand{\term}[1]{\ensuremath{\mathtt{#1}}\xspace}

\newcommand{\OldProblem}{BIC incentivized exploration\xspace}

\newcommand{\ALGG}{full-disclosure path\xspace}
\newcommand{\ALGGs}{full-disclosure paths\xspace}

\newcommand{\SubH}[1]{\mathcal{H}_{#1}} 

\newcommand{\fdpSup}{\mathrm{fd}}
\newcommand{\fdpL}{L^\fdpSup_K} 
\newcommand{\fdpP}{p^\fdpSup_K} 
\newcommand{\fdpN}[1][a]{N^\fdpSup_{K,#1}}

\newcommand{\GdT}{\fdpL} 
\newcommand{\GdP}{\fdpP} 

\newcommand{\estN}{N_{\term{est}}}
\newcommand{\estC}{C_{\term{est}}}

\newcommand{\conf}[1]{\mathtt{conf}\left(#1\right)}
\newcommand{\event}[1]{\ensuremath{\mathtt{event}_{#1}}\xspace}
\newcommand{\cG}{\mathcal{G}}

\def\2LEVEL{two-level policy}

\def\NG{\sigma}

\def\E{\mathbb{E}}
\def\reg{\mathrm{Reg}}
\def\A{\mathcal{A}}

\def\cD{\mathcal{D}}
\def\cT{\mathcal{T}}

\def\z{\tau}

\title{Incentivizing Exploration with Selective Data Disclosure%
\footnote{
A preliminary version has been accepted as a full paper at \emph{ACM EC 2020} (ACM Conference on Economics and Computation), and published as a two-page abstract in the conference proceedings.
\vspace{1mm}

\xhdr{Version History.}
A working paper has been available at {\tt https://arxiv.org/abs/1811.06026} since Nov 2018. The conference publication corresponds to the Feb'20 version. These initial versions contained all results except robustness and the numerical study (\Cref{sec:robust,sec:expts}), which were added in Dec'20 and Nov'24, respectively. New discussions (\Cref{sec:transparency,sec:discussion-model}) were added in March'26, as well as a partial reframing of the motivating story to emphasize transparency and deemphasize commitment. Also, the March'26 revision removed the additional  technical assumptions on agent behavior, leaving only Assumption~\ref{ass:embehave}.
\vspace{1mm}

\xhdr{Acknowledgements.}
We would like to thank Robert Kleinberg for a brief collaboration in the early stages of this project, and Ian Ball and Nicholas Lambert for helpful feedback. We would like to thank seminar attendees at Columbia, Michigan, MIT-Harvard, NYU, Princeton, Yale, the Simons Institute for the Theory of Computing, Stanford; and workshop attendees at the Arizona State University Economic Theory Conference, the Stony Brook Game Theory Center, and the UBC-HKU Summer Theory Conference.
\vspace{1mm}
}}

\date{First version: November 2018\\This version: March 2026}

\author{
Nicole Immorlica \thanks{Microsoft Research (Cambridge, MA) and Yale University (New Haven, CT).
Email: nicimm@microsoft.com.\newline
Research done while employed full-time at Microsoft Research.}
\and
Jieming Mao \thanks{Google, New York, NY.
Email: maojm@google.com.\newline
Research done during an internship at Microsoft Research NYC, while J.M. was a PhD student at Princeton University.}
\and
Aleksandrs Slivkins  \thanks{Microsoft Research, New York, NY.
Email: slivkins@microsoft.com.}
\and
Zhiwei Steven Wu  \thanks{Carnegie-Mellon University, Pittsburgh, PA.
Email: zstevenwu@cmu.edu.\newline
Research done during a postdoc at Microsoft Research NYC.}
}

\begin{document}

\maketitle

\begin{abstract}
We propose and design recommendation systems that incentivize efficient exploration.  Agents arrive sequentially, choose actions and receive rewards, drawn from fixed but unknown action-specific distributions. The recommendation system presents each agent with actions and rewards from a subsequence of past agents, chosen ex ante. Thus, the agents engage in sequential social learning, moderated by these subsequences. We asymptotically attain optimal regret rate for exploration, using a flexible frequentist behavioral model and mitigating rationality 
assumptions inherent in prior work. We suggest three components of effective recommendation systems: independent focus groups, group aggregators, and interlaced information structures.

\end{abstract}

\setcounter{page}{0}
\thispagestyle{empty}

\newpage
\addtocontents{toc}{\setcounter{tocdepth}{0}}

\onehalfspacing

\section{Introduction}
\label{sec:intro}

A prominent feature of online platform markets is the pervasiveness of reviews and ratings.  
%
The review and rating ecosystem creates a dilemma for  market design. On the one hand, platforms would like to allow each consumer to make an informed choice
by presenting the most comprehensive and comprehensible information. On the other hand, platforms need to encourage consumers to explore infrequently-selected alternatives in order to learn more about them. Extensive exploration may be required in settings, like ours, where the reward of an alternative is stochastic.
The said exploration, while beneficial for the common good, is often misaligned with incentives of individual consumers. Being short-lived, individuals prefer to exploit available information, selecting alternatives that look best based on this information.
This behavior can cause herding in which all consumers take a sub-optimal alternative if, for example, all consumers see all prior ratings.
Aside from such extreme behaviors, some alternatives may get explored at a very suboptimal rate, or suffer from selection bias.
Thus platforms must incentivize exploration.

\cite{Kremer-JPE14} and \cite{Che-13} introduced the problem of incentivized exploration in the context of platform design.  Their work, along with extensive follow-up work, leverages information asymmetry to mitigate the tension between exploration and exploitation.
The platform chooses a single recommendation for each consumer based on past ratings, and does not disclose any other information about the ratings. Assuming, as is standard, that consumers are Bayesian rational,
platforms can incentivize sufficient exploration to enable efficient social learning. However, this assumption can be problematic in practice: consumers may hesitate to follow recommendations because of limited rationality, behavioral biases, or a desire for transparency stemming from a preference for detailed and interpretable information.


Our work also leverages information asymmetry to induce social learning, but does so with a restricted class of platform policies which enable a more permissive behavioral model.
We restrict the platform to
delivering messages, which we call order-based disclosure policies, which provide each consumer with a subhistory of past choices and ratings. Specifically, a partial order on the arrivals is fixed ex-ante (and can be made public w.l.o.g.), and each consumer observes the chosen alternatives and the ratings for everyone who precedes her in this partial order. Put differently, an order-based disclosure policy constructs a communication network for the consumers, and lets them engage in social learning on this network. We assume consumers act like frequentists:
to estimate the reward of a given alternative, they follow the empirical mean of past ratings and form a confidence interval. The actual estimate can lie in a wide range consistent with the confidence interval.%
\footnote{Our frequentist model encompasses Bayesian agents with well-specified Beta-Bernoulli beliefs (see the last paragraph of Section~\ref{sec:model-agents}), and allows substantial deviations from Bayesian rationality.}
This is justified because each provided subhistory is unbiased -- it cannot be biased to make a particular action look good as it is chosen ex ante -- and transitive -- it contains the information sets of all consumers therein.

Our framework provides several key benefits in terms of transparency and rationality.
First, due to the unbiasedness and transitivity properties mentioned above, the only ratings that can possibly influence beliefs of a given consumer  are those included in her subhistory. This provides the maximal possible transparency, in the sense that the consumer is presented with full history of the mechanism relevant to this consumer.%
\footnote{This point, as well as the next point regarding taking data at face value, are made formal in \Cref{sec:transparency}.}
Second, the subhistory can be taken ``at face value", as if the chosen alternatives were fixed ahead of time.  This simplifies consumer's reasoning: in particular, she does not need to reason about the strategic rationale for the observed prior choices. While in prior work on incentived exploration the consumer had to interpret the recommendation in light of the recommendation policy, this type of strategic reasoning is longer needed in our model.
Third, a simplified model of rationality suffices: we only need to define how consumers react to full history, taking data ``at face value". We do \emph{not} need to specify how consumers reason about other consumers under incomplete data.
Fourth, we only need an agent to understand that she is given some subhistory which is unbiased and transitive. It does not matter to the agent what subset of arrivals is covered by this subhistory, and how it is related to the other agents' subsets.
Finally, the consumer model can be flexible. The consumers can deviate from exact utility-maximization and exhibit considerable behavioral biases. Their biases, beliefs and preferences can differ from one consumer to another, and need not be known to the platform.

We design several order-based disclosure policies in the context of this framework, of increasing complexity and improving performance guarantees.  Our policies intertwine subhistories in a certain way, provably providing consumers with enough information to converge on the optimal alternative.  Our best policy matches the best possible convergence rates, even absent incentive constraints. This policy also ensures that each consumer sees a substantial fraction of the history.

Our work suggests the importance of several design considerations.
First, independent focus groups provide natural exploration due to random fluctuations in observed rewards.  These natural experiments can then be provided to future consumers.
Second, improving beyond very suboptimal learning rates requires adaptive exploration which gradually zooms in on the better alternatives.
For example, if the focus groups learn the optimal alternative quickly, then this information should be propagated; otherwise additional exploration is required. This adaptivity can be achieved, even with subhistories chosen ex ante, by introducing group aggregators that see the subhistory of some, but not all, focus groups.  Third, optimal learning rates require reusing observations; otherwise too many consumers make choices with limited information.  The reused observations must be carefully interlaced to avoid contamination between experiments.





We start with a simple policy which runs a full-disclosure policy in parallel on several disjoint subsets of consumers (the focus groups mentioned above), collects all data from these runs, and discloses it to all remaining consumers.%
\footnote{The full-disclosure policy on a given subset $S$ of consumers reveals to each consumer in $S$ the history of ratings for all previous consumers in $S$.}
 We think of this policy  as having two levels: Level 1 contains the parallel runs, and Level 2 is everyone else (corresponding to exploration and exploitation, respectively). While this policy provably avoids herding on a suboptimal alternative, it over-explores bad alternatives and/or under-explores the good-but-suboptimal ones, which makes for very inefficient learning.

Our next step is a proof-of-concept implementation of adaptive exploration, achieving a proof-of-concept improvement over the previous construction.
We focus on the case of two alternatives, and upgrade the simple two-level policy with a middle level. Each consumer in this new level is a
group aggregator who receives the data collected by its respective group:  a subset of parallel runs from the first level.
These consumers explore only if the gap between the best and second-best alternative is sufficiently small, and exploit otherwise.  When the gap is small, the parallel runs do not have sufficient time to distinguish the two alternatives before herding on one of them.  However, a given alternative can, with some probability, empirically outperform the others within in a given group, inducing the group aggregator to explore it more. This provides the third-level consumers with enough samples to distinguish the two alternatives.



The main result essentially captures the full benefits of adaptive exploration.
We extend the three-level construction to multiple levels, connected in fairly intricate ways, using group aggregators and reusing observations as discussed above. For each piece of our construction, we prove that consumers' collective self-interested behavior guarantees a certain additional amount of exploration if, and only if, more exploration is needed at this point. The  guarantee substantially depends on the parameters of the problem instance (and on the level at which this piece resides), and critically relies on how the pieces are wired together.

Our framework is directly linked to multi-armed bandits, a popular abstraction for designing algorithms to balance exploration and exploitation. An order-based policy incentivizes consumers to implement some bandit algorithm, and consumers' welfare is precisely the total reward of this algorithm. The two-level policy implements a well-known bandit algorithm called explore-then-commit, which explores in a pre-defined way for a pre-set number of rounds, then picks one alternative for exploitation and stays with it for the remaining rounds. Our multi-level policy implements a ``batched" bandit algorithm which can change its exploration schedule only a small number of times, each change-point corresponding to a level in our construction.


We analyze our policies in terms of regret, a standard notion from the literature on multi-armed bandits, defined as the difference in total expected rewards between the best alternative and the algorithm.%
\footnote{Essentially, this is how much one regrets not knowing the best arm in advance.} We obtain regret rates that are sublinear in the time horizon, implying that the average expected reward converges to that of the best alternative. The multi-level policy matches the optimal regret rates for bandits,
for a constant number of alternatives.
The two-level policy matches the standard (and very suboptimal) regret rates of bandit algorithms such as explore-then-commit that do not use adaptive exploration. And the three-level policy admits an intermediate  guarantee.
Moreover,
regret bounds for the multi-level policy decrease drastically for easy instances of multi-armed bandits where the marginal benefit of the best option is high, a qualitative improvement compared to the two- and three-level policies.

Our performance guarantees are robust in that they hold in the worst case over a class of reward distributions, and do not rely on priors. Moreover, our constructions are robust to small amounts of misspecification. First, all parameters can be increased by at most a constant factor (and the two-level construction allows a much larger amount of tweaking). Second, we accommodate some information leakage, \eg rounds that are observable by other focus groups.

While this paper is mainly theoretical, we conduct a limited-scope numerical study to assess feasibility of our approach, focusing on the first level of our constructions (and with clear implications for the two-level construction).

\xhdr{Map of the paper.}
\Cref{sec:related-work} surveys related work. \Cref{sec:model-prelims} presents the model of incentivized exploration and our approach within this model (\ie order-based policies with frequentist agents). The next three sections present our results on order-based policies, progressing from two to three to multiple levels as discussed above. \Cref{sec:robust,sec:expts} are, resp., on robustness and the numerical study.
All proofs are deferred to the appendix, as well as additional discussions of some prior work (\Cref{sec:discussion}) and some realistic concerns beyond the scope of our model (\Cref{sec:discussion-model}).

\section{Related work}
\label{sec:related-work}
The problem of incentivized exploration, as studied in our paper, was introduced in \citet{Kremer-JPE14} and was motivated, as is our work, by recommendation systems.\footnote{In a simultaneous and independent work, \citet{Che-13} formalize a similar motivation using a very different model with continuous information flow and a continuum of agents.}
Similar to this literature, our work features short-lived agents whose actions are coordinated by a central principal with an eye towards maximizing long-run welfare.
However, the models in prior work required a strong assumption of Bayesian rationality,
even faced with complex, non-transparent policies. Under this assumption, a platform's policy can be reduced to a multi-armed bandit algorithm which recommends an action to each agent and satisfies Bayesian incentive-compatibility (\emph{BIC}).%
\footnote{We call this problem \emph{\OldProblem}. Various facets of this problem have been investigated:
optimal solutions for deterministic rewards
    \citep[and two arms:][]{Kremer-JPE14},
regret-minimization for stochastic rewards
    \citep[and many arms:][]{ICexploration-ec15,Selke-PoIE-ec21},
exploration-maximization for heterogenous agents
    \citep{ICexplorationGames-ec16,Jieming-multitypes18},
information leakages
    \citep[for deterministic rewards and two arms:][]{Bahar-ec16,Bahar-ec19},
large structured action sets and correlated priors
    \citep{IncentivizedRL,CombiIE-neurips22},
and time-discounted rewards
    \citep{Bimpikis-exploration-ms17}.
Surveys of this work can be found in \citet[Chapter 11]{slivkins-MABbook} and \cite{IncentivizedExploration-chapter}. A version of BIC incentivized exploration with monetary incentives but without information asymmetry  was studied in \citet{Frazier-ec14,Kempe-colt18}. }
Thus, the agents either need to trust the BIC property or to verify it. The former is arguably a lot to take on faith, and the latter typically requires a detailed knowledge of the algorithm and a sophisticated Bayesian reasoning. Moreover, agents may be irrationally averse to recommendations without any supporting information, or to the possibility of being singled out for exploration.

The novelty in our work is twofold: (i) the principal is restricted to disclosing past reviews (whereas prior work permitted an arbitrary message space, and w.l.o.g. focused on direct recommendations), and (ii) the agents are allowed a broad class of data-driven behaviors.
%
We design policies that incorporate these realistic assumptions without sacrificing performance, and illuminate novel insights that go beyond prior work. As discussed in the Introduction, we find that recommendation systems benefit from partitioning early agents into ``focus groups;'', providing subsequent agents with information gathered by ``independent'' focus groups, so as to avoid herding; carefully ``reuse'' this information to speed up the learning process. 

Our restriction to disclosure policies has precedence in the literature on strategic disclosure initiated by \cite{grossman1981informational,milgrom1981good,milgrom1986relying}. As in those models, our paper assumes only a selection of facts can be disclosed to the potential consumer, and those facts (other consumers' reviews in our case) can not be manipulated by the platform.  Those papers often exhibit unraveling due to the ability of the sender to select which facts to disclose based on the facts themselves.  In contrast, our policies commit to a selection of reviews up front, before they are revealed, and so avoid unraveling.

Prior work has also considered full disclosure, especially in the context of recommendation systems, in which every prior review is revealed to each consumer. A full-disclosure policy implements the ``greedy" bandit algorithm which only exploits, and suffers from herding on a suboptimal alternative
(see Section~\ref{sec:discussion}).
However,
under strong assumptions on
the primitives of the economic environment, including
the structure of rewards and diversity of agent types,
full disclosure avoids herding and performs well for heterogenous agents  \citep{kannan2018smoothed,bastani2017exploiting,externalities-colt18,AcemogluMMO19}.%
\footnote{Most of these results use a mathematically equivalent framing in terms of multi-armed bandits.}
Our work avoids herding by design, leveraging partial disclosure policies that guarantee a degree of ``independence'' between information flows. \cite{vellodi2018ratings} similarly observes that suppressing reviews can improve recommendation systems relative to full disclosure policies, albeit due to endogenous entry of firms.

Our behavioral assumptions have roots in prior work on non-Bayesian models of behavior.
In much of this literature, agents use (often naive) variants of statistical inference
which infer the world state from samples,
\eg ``case-based decision theory" of \citet{CaseBased-qje95}.
Our frequentist agents similarly rely on simple forms of data aggregation to form beliefs, albeit with good justification as the subhistories they observe are unbiased and transitive.%
\footnote{Our model of frequentist agents is technically a special case of that in \citet{CaseBased-qje95}.
We note that our model also admits Bayesian-rational agents with certain priors, as discussed in Section~\ref{sec:model-agents}.}
Such behavioral models are prominent in the literature on social learning, starting from \citet{DeGroot74}.
They are well-founded in experiments, as good predictors of human behavior in some social learning scenarios
\citep{chandrasekhar2020testing,dasaratha2019experiment}.
The behaviors we study are also reminiscent of the inference procedures studied in \cite{salant2020statistical} and the learning algorithms analyzed in \cite{cho2020machine,liang2019games}, albeit in different settings.

Our work can be interpreted as coordinating social learning (by designing a network on which the social learning happens). However, all prior work on social learning studies models that are very different from ours, including a variant of sequential social learning. We defer the detailed comparison to
Section~\ref{sec:discussion}. Interestingly, \cite{dasaratha2020learning} optimize the social network for the sequential variant mentioned above, under ``naive" agents' behavior, and observe that silo structures akin to our two-level policy improve learning rates.

Our perspective of multi-armed bandits is very standard in machine learning theory -- the primary community where bandit algorithms are designed and studied over the past 2-3 decades -- but perhaps less standard in economics and operations research. In particular, algorithms are designed for vanishing regret without time-discounting (rather than Bayesian-optimal time-discounted reward, a more standard economic perspective), and compared theoretically based on their asymptotic regret rates. A key distinction emphasized in the machine-learning literature as well as in our paper is whether the exploration schedule is fixed in advance or optimally adapted to past observations. The vast literature on regret-minimizing bandits is summarized in
\citep{Bubeck-survey12,slivkins-MABbook,LS19bandit-book}. The social-planner version of our model corresponds to \emph{stochastic bandits}, a standard, basic version with i.i.d. rewards and no auxiliary structure. ``Batched" version of stochastic bandits (which underpins our multi-level construction) was introduced in \citet{Perchet2015BatchedBP} and subsequently studied, \eg
in \citet{BatchedBantits-neurips19,BatchedBandits-aaai21,BatchedBandits-neurips20}.
Markovian, time-discounted bandit formulations \citep{Gittins-book11,Bergemann-survey06} and various other connections between bandits and mechanism design (surveyed, \eg in \citet[Chapter 11.7]{slivkins-MABbook}) are less relevant to our model.



\OMIT{ 
The performance metric in prior work, as in ours, can be cast as regret. \nicomment{not really sure about the econ ones...}
The problem of optimizing regret for \emph{\OldProblem} was largely resolved in \cite{Kremer-JPE14} and the subsequent work \cite{ICexploration-ec15,ICexplorationGames-ec16} under an assumption on the number of actions $K$. For $K=2$ actions \cite[\eg][]{Kremer-JPE14,Che-13,Bimpikis-exploration-ms17,Bahar-ec16}, or a constant number of actions $K$ \citep[\eg][]{ICexploration-ec15,ICexplorationGames-ec16}, these papers achieve the optimal regret rate (see Equation~\eqref{eq:model-OptRegret}) for bandit algorithms without incentives. \niedit{\footnote{In contrast to our model, much of this prior work restricts attention to deterministic rewards, e.g., \citep[\eg][]{Kremer-JPE14,Bahar-ec16}.}} \nicomment{does che-horner have stochastic rewards?} The results involves a multiplicative ``constant" that can get arbitrarily large depending on the Bayesian prior.
Our result achieves a similar regret rate and similarly depends on a parameter in our choice model. The regret bounds in the prior work, as well as ours, scale exponentially in $K$.\footnote{In comparison, bandit algorithms without incentives achieve regret rates that scale as $\sqrt{K}$.} Thus, up to this exponential dependence on $K$ also present in prior work, our policies exhibit optimal performance.\footnote{In work subsequent to ours, \citet{Selke-PoIE-ec21} achieve $\poly(K)$ regret scaling, albeit only when the Bayesian prior is independent across arms and only for Bayesian regret (\ie regret in expectation over the Bayesian prior).}
} 

\section{Our model and approach}
\label{sec:model-prelims}


We study the incentivized exploration problem, in which a platform (\emph{principal}) faces a sequence of $T$ myopic consumers (\emph{agents}). There is a set $\A$ of possible actions (\emph{arms}). At each round $t\in [T] := \{1,2 \LDOTS T\}$, a new agent $t$ arrives, receives a \emph{message} $m_t$ from the principal, chooses an arm $a_t\in \A$, and collects a binary reward $r_t\in \{0,1\}$.
The message $m_t$ could be arbitrary, \eg a recommended action or, in our case, a subset of past reviews. The principal chooses messages according to a decision rule called the \emph{disclosure policy}. The agents' response $a_t$ is defined as a function of round $t$ and message $m_t$. Both the disclosure policy and agent behavior are constrained to a subset of well-motivated choices, as per Sections~\ref{sec:model-policies} and~\ref{sec:model-agents}.
The reward from pulling an arm $a\in \A$ is drawn independently from Bernoulli distribution $\cD_a$ with an unknown mean reward $\mu_a$. A problem instance is defined by (known) parameters $|\A|$, $T$ and (unknown) mean rewards $\mu_a:\,a\in\A$.

The information structure is as follows. Each agent $t$ does not observe anything from the previous rounds, other than the message $m_t$. The chosen arm $a_t$ and reward $r_t$ are observed by the principal (which corresponds, \eg, to the consumer leaving a rating or review on the platform).

We assume that mean rewards are bounded away from $0$ and $1$, to ensure sufficient entropy in rewards. For concreteness, we posit $\mu_a\in [\tfrac13,\tfrac23]$.

While we focus on the paradigmatic case of Bernoulli rewards, we can handle arbitrary rewards $r_t \in [0,1]$ with only minor modifications to the analysis.\footnote{We can round each reward $r_t$ up or down using an independent Bernoulli draw with mean $r_t$, and then using these ``rounded rewards" instead of the true ones. This corresponds to binary ratings: thumbs up or thumbs down. Alternatively, we can assume that the reward distribution for each arm places at least a positive-constant probability on (say) subintervals $[0,\nicefrac14]$ and $[\nicefrac34,1]$.}
In essence, the range of rewards is small compared to the number of samples, like in all prior work on incentivized exploration. This is a very standard assumption throughout machine learning, and it is justified in small-stakes applications such as recommendation systems for movies, restaurants, etc.
We could also handle sub-Gaussian reward distributions with variance $\leq 1$ in a similar manner. For arbitrary reward distributions with support $[0,R]$, regret bounds scale linearly in $R$.


\subsection{Objective: regret}
\label{sec:prelims-MAB}

The principal's objective is to maximize agents' rewards. The social-planner version, when the principal can directly choose actions $a_t$ without any restrictions, is precisely the basic version of multi-armed bandits termed \emph{stochastic bandits}. We characterize the principal's performance using the notion of \emph{regret}, a standard objective in stochastic bandits. Formally, regret is defined as
\begin{align}\label{eq:intro-regret-defn}
\textstyle
  \reg(T)
  = T \max_{a\in \A} \mu_a -
  \sum_{t\in [T]} \E[\mu_{a_t}],
\end{align}
where the expectation is over the randomness in rewards and the messaging policy.  Thus, regret is the difference, in terms of the total expected reward, between the principal's policy and the first-best policy which knows the mean rewards a priori. Notably, the rewards in this objective are not discounted over time.%
\footnote{This is a predominant modeling choice in the literature on bandits over the past 20+ years, as well as in most prior work on incentivized exploration. A standard motivation is that the algorithm/mechanism sees many users in a relatively short time period.}

Following the bandit literature, we focus on the dependence on $T$, the number of agents (which is, effectively, the time horizon). Assuming regret is sublinear in $T$, the average expected reward converges to that of the best arm at rate $\reg(T)/T$. We are mainly interested in robust upper bounds on regret that hold \emph{in the worst case} over all (valid) mean rewards. This provides guarantees (even) for a principal that has no access to a prior or simply does not make use of one due to extreme risk aversion.
We are also interested in performance of a policy at a given round $t$, as measured by \emph{instantaneous regret}
    $\max_{a\in \A} \mu_a - \E[\mu_{a_t}]$,
also known as \emph{simple regret}.%
\footnote{Note that summing up the simple regret over all rounds $t\in[T]$ gives $\reg(T)$.}

\begin{remark}
Our regret bounds are stated in terms of expected regret, as per \refeq{eq:intro-regret-defn}. Our analysis also yields similar regret bounds with high probability, namely bounds on pseudo-regret
$T \max_{a\in \A} \mu_a -
  \sum_{t\in [T]} \mu_{a_t}$
that hold with probability at least $1-T^{-c}$. This can be achieved for any fixed absolute-constant $c>0$ via a minor modification in absolute constants in our analysis.
\end{remark}



Regret in our model can be directly compared to regret in the stochastic bandit problem with the same mean rewards.
Following the literature, we define the \emph{gap parameter} $\Delta$ as the difference between the largest and second largest mean rewards (informally, the difference in quality between the top two options). The gap parameter is not known (to the platform in incentivized exploration, or to the algorithm in bandits). Large $\Delta$, \ie the best option being far better than the second-best, naturally corresponds to ``easy" problem instances. The literature is mainly concerned with asymptotic upper bounds.%
\footnote{We use standard asymptotic notation to characterize regret rates: $O(f(T))$ and $\Omega(f(T))$ mean, resp., at most and at least $f(T)$, up to constant factors, for large enough $T$. Similarly, $\tildeO(f(T))$ notation suppresses $\polylog(T)$ factors.}
on regret in terms of the time horizon $T$, as well as parameters $\Delta$ and the number of arms $K$.
Throughout, we assume that the number of arms $K=|\A|$ is constant. However, we explicitly note the dependence on $K$ when appropriate, \eg we use $O_K(\cdot)$ notation to note that the ``constant" in $O()$ can depend on $K$ (and nothing else).

Optimal regret rates in stochastic bandits
\citep{bandits-ucb1,bandits-exp3,Lai-Robbins-85} are
\begin{align}\label{eq:model-OptRegret}
\reg(T) \leq O\rbr{\min\rbr{
    \sqrt{KT\log T},\; \tfrac{K}{\Delta} \log T
    }}.
\end{align}
This includes a \emph{worst-case} regret rate $O(\sqrt{KT\log T})$ which applies to all problem instances, and a \emph{gap-dependent} regret rate of
    $O(\tfrac{K}{\Delta} \log T)$.
We match both regret rates for a constant number of arms. Either regret rate can only be achieved via \emph{adaptive exploration}: \ie when the exploration schedule is adapted to the observations. It is particularly notable that we implement adaptive exploration in our environment, even though the platform must commit to the subhistories ex ante. This constraint requires a substantial number of agents to naturally serve an exploration role in some outcomes and an exploitation role in others.

A simple example of \emph{non}-adaptive exploration is the \emph{explore-then-commit} algorithm which samples arms uniformly at random for the first $N$ rounds, for some pre-set number $N$, then chooses one arm and sticks with it till the end. We implement this algorithm in our economic environment with our $2$-level policy outlined in Section~\ref{sec:warmup}. Such algorithms suffer from $\Omega(T^{2/3})$ regret, both in the worst case and for each problem instance.%
\footnote{More precisely, there is a tradeoff between the worst-case and per-instance performance: if the algorithm achieves regret $O(T^{\gamma})$ for all instances, for some $\gamma\in [\nicefrac23,1)$, then its regret for each instance can be no better than
    $\Omega\rbr{T^{2(1-\gamma)}}$.
The latter is $\Omega(T^{2/3})$ when $\gamma = \nicefrac23$. This result extends to a more general model of non-adaptive exploration, where each round either gives up on exploitation (namely: the chosen arm does not depend on the previous observations), or does not contribute to exploration (namely: its reward cannot be used in the future). This result is from \citet{MechMAB-ec09}, but the worst-case lower bound has been ``folklore knowledge" in the community.}


Regret in incentivized exploration is subject to two important limitations, even under the BIC model from prior work. \citet{ICexploration-ec15} matches \eqref{eq:model-OptRegret} for a constant number of arms $K$, but with (i) a multiplicative factor that can get arbitrarily large depending on the prior, and (ii) an exponential dependence on $K$. Both limitations are essentially inevitable \citep{Selke-PoIE-ec21}.%
\footnote{Subsequently to the conference version of our paper, \citet{Selke-PoIE-ec21} achieve $\poly(K)$ scaling in regret, albeit only when the Bayesian prior is independent across arms, only in expectation over the prior, and only under a substantial technical assumption on the family of feasible priors.}
Our result matches \eqref{eq:model-OptRegret} in a similar way (dependence on the prior is replaced with that on a parameter in the agents' choice model). We also note that much of the prior work on incentivized exploration targets $K=2$ actions \cite[\eg][]{Kremer-JPE14,Che-13,Bimpikis-exploration-ms17,Bahar-ec16}.

\subsection{Disclosure policies: order-based}
\label{sec:model-policies}
We focus on messaging policies of a particular form (and achieve optimal regret rate, in terms of $T$, despite these limitations).
First, we use \emph{disclosure policies}, where the message $m_t$ in each round $t$ discloses the \emph{subhistory} for some subset $S = S_t$ of past rounds. Formally, the subhistory is defined as
    $\SubH{S} = \cbr{ (s,a_s,r_s):\;s\in S }$,
where the tuple $(s,a_s,r_s)$ is the \emph{outcome} for a given agent $s\in S$.
The subhistory can correspond, for example, to a subset of past reviews.
Second, we assume that the subset $S_t$ is chosen ex ante, before round $1$, and therefore does not depend on the previous observations. Such a message is called \emph{unbiased subhistory}; it means the platform can not bias the set of reviews it shows a consumer, \eg by selecting only those in which a particular arm has positive ratings.
Third, we fix a partial order on the rounds, and define each $S_t$ as the set of all rounds that precede $t$ in the partial order. The resulting disclosure policy is called \emph{order-based}.


Order-based disclosure policies are \emph{transitive}, in the following sense:
\begin{align}\label{eq:transitivity}
 t\in S_{t'} \Rightarrow S_t\subset S_{t'}
    \qquad \text{for all rounds $t,t'\in [T]$}.
\end{align}
In words, if agent $t'$ observes the outcome for some previous agent $t$, then she observes the entire message revealed to that agent. In particular, agent $t'$ does not need to second-guess which message has caused agent $t$ to choose action $a_t$.

For convenience, we will represent an order-based policy as an undirected graph, where nodes correspond to rounds, and any two rounds $t<t'$ are connected if and only if $t\in S_{t'}$ and there is no intermediate round $t''$ with
    $t\in S_{t''}$ and $t''\in S_{t'}$.
This graph is henceforth called the \emph{information flow graph} of the policy, or \emph{info-graph} for short. (As an illustration, see Figure~\ref{fig:path} below.) We assume that this graph is common knowledge.

\begin{figure}[h]
\begin{center}
\includegraphics[height=.8cm]{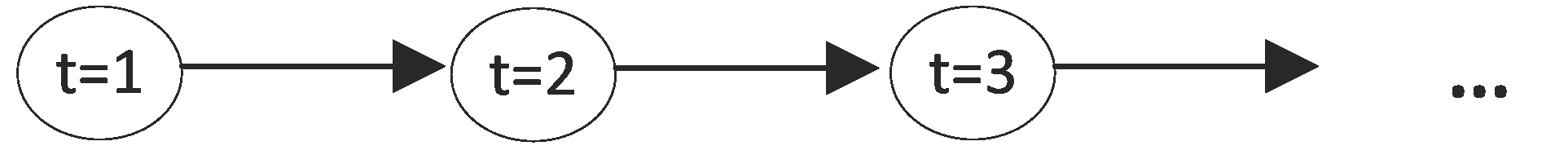}
\vspace{-5mm}
\end{center}
\caption{The information flow graph for a full disclosure policy.}
\label{fig:path}
\end{figure}

While detail-oriented agents may prefer to observe full data, our policies show all but a few past datapoints to all but a few agents,  and our main result shows a certain fraction of the full history to \emph{all} agents. Besides, even a small fraction of the full history would typically contain a large number of observations (preselected in an unbiased way), probably more than a typical agent ever needs.

\subsubsection{Discussion: Transparency}
\label{sec:transparency}

Consider the disclosure policy restricted to rounds in $S_t$, for some agent $t$, with the same subhistories as the original policy; call it the \emph{$t$-restricted policy}. This policy can be interpreted as a standalone order-based disclosure policy, since (by unbiasedness and transitivity) it cannot be affected by any agents that are not in $S_t$. In particular, all data points in the subhistories are endogenous to the restricted policy.  Moreover, this policy affects agent $t$ same way as the original mechanism. Therefore, an order-based disclosure policy can be equivalently reformulated so that each agent $t$ is subjected to the $t$-restricted policy, and sees the full history of this policy. To summarize: each agent sees the full history of a policy that she is subjected to.

Consequently, the agents are justified at taking subhistory can be taken ``at face value", as if the chosen arms recorded in the subhistory were fixed ahead of time. While this is a standard point, making it formal requires some care; we spell it out here for the sake of completeness. We endow agents with Bayesian beliefs over the mean rewards, and consider their posteriors.%
\footnote{We do not assume any specific mapping from agent posteriors to their decisions. For the rest of the paper, we only assume a frequentist model spelled out in the \Cref{sec:model-agents}, whether or not it is microfounded via Bayesian beliefs.}
Generally, the Bayesian posterior is formed given the new data and the process that generated this data. In our case, the  subhistory
    $\SubH{t}:= \SubH{S_t}$ is the new data for agent $t$,
and the data-generating process consists of the $t$-restricted disclosure policy, the random rewards, and the bandit algorithm that chooses arms. (The algorithm, in our case, consists of the agents in $S_t$ making their decisions.) Now, whatever that bandit algorithm is, it holds that the subhistory $S_t$ comprises the full history of the data-generating process. It follows that the Bayesian posterior does not depend on the bandit algorithm:

\begin{lemma}\label[lemma]{lem:face-value}
Suppose agent $t$ has a Bayesian belief over the mean rewards $(\mu_a:\,a\in\A)$
 (not necessarily the same belief as the other agents). Then the Bayesian posterior of agent $t$ given any fixed realization $H_t$ of subhistory $\SubH{t}$ is the same for any bandit algorithm such that $\Pr\sbr{\SubH{t} = H_t}>0$.
\end{lemma}

\begin{corollary}\label[corollary]{cor:face-value}
In particular, the posterior could have been generated by the ``canonical" algorithm for $H_t$ that deterministically chooses arms as recorded in $H_t$.
\end{corollary}

Thus, the formal meaning of ``taking data at face value" is expressed via
the ``canonical algorithm" from \Cref{cor:face-value}. This is a standard argument that underpins Bayesian bandit algorithms, see \citet[Ch. 3.1.2]{slivkins-MABbook} for a self-contained proof of \Cref{lem:face-value}.


\subsection{Agents' behavior: frequentist and flexible}
\label{sec:model-agents}

We assume agents behave as frequentists in response to order-based policies.
In light of the discussion above, how would a frequentist agent choose an action given the full history of observations? She would construct a confidence interval on the expected reward of each action, taking into account the average reward of this action and the number of observations, and place the action's estimate somewhere in this confidence interval. The system can provide summary statistics, so that agents would not even need to look at the raw data.

We formalize this behavior as follows. Each agent $t$ uses its observed subhistory $m_t$ to form a reward estimate $\hat{\mu}_{t,a} \in [0,1]$ for each arm $a\in \A$, and chooses an arm with a maximal estimate.%
\footnote{To simplify proofs, ties between the reward estimates are broken according to some fixed, deterministic ordering over the arms. This is to rule out adversarial manipulation of the tie breaking, and to ensure that all agents with the same data choose the same arm.}
A paradigmatic instantiation of the reward estimates is as follows:

\begin{example}\label{ex:paradigmatic}
Reward estimate $\hat{\mu}_{t,a}$ is the sample average for arm $a$ over the subhistory $m_t$, as long as it includes at least one sample for $a$; else, $\hat{\mu}_{t,a}=\tfrac12$.
\end{example}

We consider a much more permissive model, where agents can form arbitrary reward estimates as long as they lie within some ``confidence range" of the sample average. Moreover, we allow agents to have strong initial beliefs, whose effect is eventually drowned out. Formally, we make the following assumptions.


\begin{assumption}\label{ass:embehave}
Let $N_{t,a}$ and $\bar{\mu}_{t,a}$ denote the number of pulls and the empirical mean reward of arm $a$ in subhistory $m_t$. Then for some absolute constant $\estN\in \N$ and $\estC=\tfrac{1}{16}$, and for all agents $t\in [T]$ and arms $a\in\A$ it holds that
\begin{align}
\text{if}\; N_{t,a} \geq \estN
    &\quad\text{then}\quad
    \left|\hat{\mu}^t_a - \bar{\mu}^t_a \right| <
                \frac{\estC}{\sqrt{N_{t,a}}}
\label{eq:ass:embehave}\\
\text{if}\; N_{t,a} =0
    &\quad\text{then}\quad
    \hat{\mu}^t_a\geq\nicefrac13.
    \label{eq:ass:embehave-init}
\end{align}
\end{assumption}

Thus, strong(er) initial beliefs correspond to large(r) $\estN$. Note that we make no assumptions if $1\leq N_{t,a}<\estN$. The $\nicefrac13$ threshold in \refeq{eq:ass:embehave-init} can be replaced with an arbitrary strictly positive constant, with very minor changes.


\OMIT { 

We make a technical assumption: reward estimates do not depend on the ordering of the data.

\begin{assumption}\label{ass:simplifying}
In each round $t$, the estimates $(\hat{\mu}_{t,a}:\, a\in\A)$ depend only on the multiset
    $m'_t = \left\{\; (a_s,r_s):\;s\in S_t \;\right\}$,
called \emph{anonymized subhistory}.\ascomment{I forgot why we needed it!} Each agent $t$ forms its estimates according to some function $f_t$ from anonymized subhistories to $[0,1]^{|\A|}$, so that
        $(\hat{\mu}_{t,a}:\, a\in\A) = f_t(m'_t)$.
\ascomment{This is just a defn, no? So, remove from the assn?}
\end{assumption}

This function $f_t$ constitutes agent $t$'s \emph{behavioral type}. We allow multiple behavioral types, but posit that they drawn from a fixed distribution. \ascomment{Can't we just consider fixed behavioral types?? (so, remove the next assn.}

\begin{assumption}\label{ass:types}
The behavioral type of each agent $t$ is drawn independently from some fixed (but otherwise arbitrary) distribution.
\end{assumption}

} 

Several discussion points are in order:

\begin{description}

\item[(Flexibility)] Our model allows for significant flexibility in agent behavior. An optimistic (resp., pessimistic) agent may choose a reward estimate as a value towards the top (resp., bottom) of its confidence interval. Moreover, an agent can randomize its choices, by randomizing its reward estimates within their confidence intervals. This flexibility can be arm-specific as well. First, the way $\hat{\mu}_{t,a}$ depends on the data for arm $a$ can vary from one arm to another. For instance, an agent interested in restaurants can be optimistic about Chinese restaurants and pessimistic about Italian ones. Second, the reward estimates can depend on the data for other arms: \eg an agent is more optimistic about Chinese restaurants if the Italian ones are good. Third, the estimates can be arbitrarily correlated across arms: \eg it is sunny today, and an agent feels optimistic about \emph{all} restaurants.

\item[(Bayesian interpretation)]
Our model also encompasses Bayesian agents with appropriate beliefs. Specifically, suppose agents believe that each $\mu_a$ is drawn independently from some Beta-Bernoulli distribution $\mathcal{P}_a$. Then the reward estimate $\hat{\mu}_{t,a}$ is the posterior mean reward given the subhistory $m_t$. This is consistent with Assumption~\ref{ass:embehave} for a large enough prior-dependent constant $\estN$.%
\footnote{Essentially, this is because for Beta-Bernoulli priors the absolute difference between the posterior mean and the empirical mean scales as $1/\text{\#samples}$.}
Beta-Bernoulli beliefs are \emph{well-specified} in that their support necessarily contains the true model. While such beliefs are \emph{inconsistent} with the restriction that $\mu_a \in [\nicefrac{1}{3}, \nicefrac{2}{3}]$, one could argue that Bayesian agents might be unaware of this restriction. Also, our disclosure policies guarantee that the posterior beliefs are ``asymptotically consistent" with $\mu_a \in [\nicefrac{1}{3}, \nicefrac{2}{3}]$, in the sense that they get arbitrarily close to $[\nicefrac{1}{3}, \nicefrac{2}{3}]$ over time.%
\footnote{The formal statement is as follows: for some $\alpha,\beta\in(0,1)$, each agent $t>T^{\alpha}$ forms a posterior belief such that
    $\Pr\sbr{\mu_a\not\in [\nicefrac{1}{3}, \nicefrac{2}{3}]}< O(T^{-\beta})$
for each arm $a$. This is because our disclosure policies guarantee that such agents $t$ observe sufficiently many samples of each arm.}

\item[(A consistency issue)]
The fact that $\hat{\mu}_{t,a}$ falls below $\nicefrac13$ after a long sequence of low rewards appears inconsistent with the restriction that $\mu_a \in [\nicefrac{1}{3}, \nicefrac{2}{3}]$. However, one could argue that agents are unaware of this restriction because they have incomplete information and/or are unsophisticated. Alternatively, all reward estimates can be projected into the $[\nicefrac{1}{3}, \nicefrac{2}{3}]$ interval,%
\footnote{That is, truncate the reward estimate at $\nicefrac13$ if it is too low, and at $\nicefrac23$ if it is too high.} assuming random tie-breaking when multiple arms achieve the highest reward estimate. This variant works with minimal changes.

\end{description}


\section{A simple two-level policy}
\label{sec:warmup}

We first design a simple policy that exhibits asymptotic learning (i.e., sublinear regret).  While not achieving an optimal regret rate, this policy illuminates a key feature: initial agents are partitioned into \emph{focus groups}. Each agent sees the history for all previous agents in the same focus group (and nothing else).
The information generated by these focus groups is then presented to later agents. We think of this policy as having two \emph{levels}: the exploration level containing the focus groups, followed by the exploitation level. All agents in the latter observe full history.

We first describe the structure of a single focus group. Consider a disclosure policy that reveals the full history in each round $t$, \ie $S_t = [t-1]$; we call it the \emph{full-disclosure policy}. The info-graph for this policy is a simple path. Intuitively, all agents in a the path of this full-disclosure policy are in a single focus group.

\begin{definition}
A subset of rounds $S\subset [T]$ is called a \emph{full-disclosure path} in the  info-graph $G$ if the induced subgraph $G_S$ is a simple path, and it connects to the rest of the graph only through the terminal node $\max(S)$, if at all.
\end{definition}

\noindent Full-disclosure paths are useful primitives for exploration as they guarantee that each arm is sampled with a positive-constant probability. This happens due to stochastic variation in outcomes; some agents in a focus group will get uncharacteristically bad rewards from an arm, inducing others to pull a different arm. We prove that path length at least $\fdpL$ suffices to guarantee this, for some parameter $\fdpL$ that depends only on $K$, the number of arms (see Lemma~\ref{lem:greedy}). We build on this fact throughout the paper.


Our \emph{\2LEVEL} policy builds upon this primitive, allowing future agents to exploit the information that early agents explore, thereby closely following the well-studied ``explore-then-commit'' paradigm from multi-armed bandits. For a parameter $T_1$ fixed later, the first $N = T_1\cdot \fdpL$ agents comprise the ``exploration level." These agents are partitioned into $T_1$ full-disclosure paths of length $\fdpL$ each, where $\fdpL$ depends only on the number of arms $K$. In the ``exploitation level", each agent $t>N$ receives
the full history, \ie $S_t = [t-1]$.
\footnote{For the regret bounds, it suffices if each agent in the exploitation level only observes the history from the exploration level, or any superset thereof.}
The info-graph for this disclosure policy is shown in Figure~\ref{fig:2level}.

\begin{figure}[t]
\centering
\begin{tikzpicture}
 \filldraw[fill=blue!20!white]
 (0,2)--(10,2)--(10,3)--(0,3)--cycle;
  \filldraw[fill=red!20!white]
  (0,0)--(1,0)--(1,1)--(0,1)--cycle;
  \draw (0.5,1)--(5,2);
  \filldraw[fill=red!20!white]
  (1,0)--(2,0)--(2,1)--(1,1)--cycle;
  \draw (1.5,1)--(5,2);
  \filldraw[fill=red!20!white]
  (2,0)--(3,0)--(3,1)--(2,1)--cycle;
  \draw(2.5,1)--(5,2);
  \filldraw[fill=red!20!white]
  (3,0)--(4,0)--(4,1)--(3,1)--cycle;
  \draw(3.5,1)--(5,2);
  \filldraw[fill=red!20!white]
  (9,0)--(10,0)--(10,1)--(9,1)--cycle;
  \draw(9.5,1)--(5,2);
  \node at(5,0.5){$\cdots$};
  \node at(6,0.5){$\cdots$};
  \node at(7,0.5){$\cdots$};
  \node at(8,0.5){$\cdots$};
  \node at(5,2.5){all remaining rounds};
  \node at(0.5,0.5){$\fdpL$};
  \node at(1.5,0.5){$\fdpL$};
  \node at(2.5,0.5){$\fdpL$};
  \node at(3.5,0.5){$\fdpL$};
  \node at(9.5,0.5){$\fdpL$};
  \node at(-1,0.5){\textbf{Level 1}};
  \node at(-1,2.5){\textbf{Level 2}};
  \draw[->] (11,0)--(11,3);
  \node at(11.5,1.5)[ rotate=90]{Time};

  \draw [decorate,decoration={brace,amplitude=10pt},xshift=0pt,yshift=0pt] (10,-0.2) -- (0,-0.2) node [black,midway,yshift=-0.6cm]
  {$T_1$ full-disclosure paths of length $\fdpL$ each};
\end{tikzpicture}

\caption{Info-graph for the 2-level policy. }
\label{fig:2level}
\end{figure}

We show that this policy incentivizes the agents to perform non-adaptive exploration, and achieves a regret rate of  $\tilde O_K(T^{2/3})$. The key idea is that since one full-disclosure path collects one sample of a given arm with (at least) a positive-constant probability, using many full-disclosure paths ``in parallel" ensures that sufficiently many samples of this arm are collected with very high probability. The proof of the following theorem can be found in \Cref{sec:pfs-2level}.

\begin{theorem}\label{thm:2level}
The \2LEVEL with parameter $T_1 = T^{2/3}\,(\log T)^{1/3}$ achieves regret
\[ \reg(T) \leq O_K\left( T^{2/3}\, (\log T)^{1/3} \right).\]
\end{theorem}

\begin{remark}
All agents $t>T^{2/3}\,(\log T)^{1/3}$ (\ie all but the vanishingly small fraction of agents who are in the exploration level) observe full history, and pull an arm with instantaneous regret at most $\tildeO\rbr{T^{-1/3}}$.
\end{remark}

\begin{remark}\label{rem:2level-tweaks}
Each full-disclosure path can be made arbitrarily longer, and more full-disclosure paths can be added (of arbitrarily length), as long as the total number of Level-$1$ agents increases by at most a constant factor. Same regret bounds are attained with minimal changes in the analysis.
\end{remark}

\begin{remark}
For a constant $K$, the number of arms, we match the optimal regret rate for non-adaptive multi-armed bandit algorithms. If the gap parameter $\Delta$ is known to the principal, then (for an appropriate tuning of parameter $T_1$) we can achieve regret
  $\reg(T) \leq O_K(\log(T) \cdot \Delta^{-2})$.
\end{remark}

One important quantity is the expected number of samples of a given arm $a$ collected by a full-disclosure path $S$ of length $\fdpL$, \ie the number of times this arm appears in the subhistory $\SubH{S}$. Indeed, this number, denoted $\fdpN$, is the same for all such paths. We use this quantity, here and in the subsequent sections, through a concentration inequality which aggregates the effect of having multiple full-disclosure paths (see Lemma~\ref{lem:t1runs}).

\subsection{Counterexamples}

Although simple, the two-level policy does exhibit some subtleties. First, it is important that the focus groups are independent. For example, a few initial agents observable by everyone may induce herding on a suboptimal arm. This might happen if, for example, the initial agents are celebrities, and their experiences leak to future agents outside the platform. Essentially, these ``celebrities" herd on a suboptimal arm with constant probability (because they observe each other), and this herding persists afterwards (because everyone else observes the ``celebrities"). We flesh out this point in the following example, analyzed in the online appendix.

\begin{example}\label{ex:robust-global}
Posit $K=2$ arms  such that
    $\nicefrac{3}{4} \geq \mu_1>\mu_2>\nicefrac{1}{4}$.
Suppose Assumption~\ref{ass:embehave} holds with $\estN=2$ so that each  arm is chosen in the first two rounds, and subsequently the mean reward of each arm $a$ is estimated by the sample average (\ie
    $\hat{\mu}^t_a := \bar{\mu}^t_a $
for all rounds $t>2$). If each of the first $R$ rounds are observable by all subsequent agents,  for a large enough $R = \Omega(\sqrt{\log(T)})$, then with (at least) a positive-constant probability it holds that all agents $t>R$ choose arm $2$.
\end{example}

Second, it is important that each focus group has a linear information flow. For example, the first few agents acting in isolation may force high-probability herding within the focus group, preventing the natural exploration that we rely on. This may happen, for example, if their reviews are submitted and/or processed with a substantial delay. Suppose these initial agents are pessimistic about arm $1$, so that each one in isolation pulls arm $2$. This builds certainty about the mean reward of arm $2$ which, for an appropriate setting of parameters, may exceed the initial reward estimate for arm $1$.  Then later agents viewing all this information will, with high probability, fail to pull arm $1$ (even though arm $1$ may be optimal).

\begin{example}\label{ex:robust-local}
Suppose there are only two arms, all agents initially prefer arm $2$, and have the same initial reward estimate $\hat{\mu}_1$ for arm $1$. Consider a full-disclosure path $P$ starting at round $t_0$. Suppose agent $t_0$ observes $R$ ``leaf agents" (each of which does not observe anybody else). Then, for any absolute constant $\mu_2 >\hat{\mu}_1$ and a sufficiently large $R = \Omega(\sqrt{\log(T)})$, each agent in $P$ will not try arm $1$ with probability, say, at least $1-O(T^{-2})$.
\end{example}

Third, it is important that there are enough focus groups and agents therein, but not too many. Indeed, we need enough agents in each focus group to overpower the initial biases (as expressed by reward estimators with $<\estN$ samples). Having enough focus groups ensures that the natural exploration succeeds. However, agents in the focus groups would have limited information and may make suboptimal choices, so having too many of them would induce high regret.



\section{Adaptive exploration with a three-level policy}
\label{sec:3level}

The two-level policy from the previous section implements the explore-then-commit paradigm using a basic design with focus groups. The next challenge is to implement \emph{adaptive exploration}, and go below the $T^{2/3}$ barrier. Standard multi-armed bandit algorithms achieve this by pulling sub-optimal arm on occasion, when and if the available information requires it. However, we can not adaptively add focus groups since we must fix our policy ahead of time.

Instead, we accomplish adaptive exploration using a construction that adds a middle level to the info-graph. Agents in this middle level are partitioned into subgroups, each responsible for aggregating information from a subset of focus groups; we call these agents \emph{group aggregators}. For simplicity, we assume $K=2$ arms. When one arm is much better than the other,
group aggregators have enough  information to discern it and \emph{exploit}.  However, when the two arms are close, group aggregators will be induced to pull different arms (depending on the outcomes in their particular focus groups), which induces additional exploration. This construction also provides intuition for the main result, the multi-level construction presented in the next section.



\begin{construction}
The \emph{three-level policy} is an order-based disclosure policy defined as follows. The info-graph consists of three levels: the first two correspond to \emph{exploration}, and the third implements \emph{exploitation}. Like in the two-level policy, the first level consists of multiple full-disclosure paths of length $\fdpL$ each, and each agent $t$ in the exploitation level sees full history (see Figure~\ref{fig:3level}).
\footnote{It suffices for the regret bounds if each agent in the exploitation level only observes the history from exploration (\ie from all agents in the first two levels), or any superset thereof.}

The middle level consists of $\NG$ disjoint subsets of $T_2$ agents each, called \emph{second-level groups}. All nodes in a given second-level group $G$ are connected to the same nodes outside of $G$, but not to one another.



The full-disclosure paths in the first level are also split into $\NG$ disjoint subsets, called \emph{first-level groups}. Each first-level group consists of $T_1$ full-disclosure paths, for the total of $T_1\cdot \NG\cdot \fdpL$ rounds in the first layer. There is a 1-1 correspondence between first-level groups $G$ and second-level groups $G'$, whereby each agent in $G'$ observes the full history from the corresponding group $G$. More formally, agent in $G'$ is connected to the last node of each full-disclosure path in $G$. In other words, this agent receives message
    $\SubH{S}$,
where $S$ is the set of all rounds in $G$.
\end{construction}

\begin{figure}[t]
\centering
\begin{tikzpicture}
 \filldraw[fill=green!20!white]
 (0,4)--(10,4)--(10,5)--(0,5)--cycle;
 \foreach \x in {0,3,8}
 {
 \filldraw[fill=blue!20!white]
 (\x+0,2)--(\x+2,2)--(\x+2,3)--(\x+0,3)--cycle;
 \draw (\x+1,3)--(5,4);
 \filldraw[fill=red!20!white]
 (\x+0,0)--(\x+2,0)--(\x+2,1)--(\x+0,1)--cycle;
 \draw (\x+1,1)--(\x+1,2);
 \draw(\x+0.33,1)--(\x+1,2);
 \draw(\x+1.66,1)--(\x+1,2);
 \draw(\x+0.66,0)--(\x+0.66,1);
 \draw(\x+1.33,0)--(\x+1.33,1);
 \node at(\x+1,2.5){$T_2$ rounds};
 \node at(\x+0.33, 0.5){$\GdT$};
 \node at(\x+1.0, 0.5){$\cdots$};
 \node at(\x+1.66, 0.5){$\GdT$};
 \draw [decorate,decoration={brace,amplitude=10pt},xshift=0pt,yshift=0pt] (\x+2,-0.2) -- (\x+0,-0.2) node [black,midway,yshift=-0.6cm] {$T_1$ paths};
 }
  \node at(5,4.5){all remaining rounds};
  \node at (6,0.5){$\cdots$};
  \node at (7,0.5){$\cdots$};
  \node at (6,2.5){$\cdots$};
  \node at (7,2.5){$\cdots$};
  \node at(-1,0.5){\textbf{Level 1}};
  \node at(-1,2.5){\textbf{Level 2}};
  \node at(-1,4.5){\textbf{Level 3}};
  \draw[->] (11,0)--(11,5);
  \node at(11.5,2.5)[ rotate=90]{Time};

  \draw [decorate,decoration={brace,amplitude=10pt,aspect=0.33},xshift=0pt,yshift=0pt] (10,1.8) -- (0,1.8) node [black,pos=0.33,xshift = 0cm,yshift=-0.6cm] {$\NG$ groups};

\end{tikzpicture}
\caption{Info-graph for the three-level policy. Each red box in level 1 corresponds to $T_1$ full-disclosure paths of length $\GdT$ each.}
\label{fig:3level}
\end{figure}

In more detail, the key idea is as follows. Consider the gap parameter $\Delta = |\mu_1-\mu_2|$. If it is is large, then each first-level group produces enough data to determine the best arm with high confidence, and so each agent in the upper levels chooses the best arm. If $\Delta$ is small, then due to \emph{anti-concentration} each arm gets ``lucky" within  at least once first-level group, in the sense that it appears much better than the other arm based on the data collected in this group. Then this arm gets explored by the corresponding second-level group.
To summarize, the middle level exploits if the gap parameter is large, and provides some more exploration if it is small.

\begin{theorem}
\label{thm:3level}
For two arms, the three-level policy
achieves regret
\[ \reg(T) \leq O\left( T^{4/7}\, \log T \right).\]
This holds for parameters
    $T_1 = T^{4/7}\log^{-1/7}(T)$,
    $\NG = 2^{10}\log(T)$, and
    $T_2 = T^{6/7}\log^{-5/7}(T)$.
\end{theorem}

\begin{remark}
All agents $t>\tildeO(T^{6/7})$ (\ie all but a vanishingly small fraction of agents who are in the first two levels) observe full history, and pull an arm with instantaneous regret $\tildeO\rbr{T^{-3/7}}$.
\end{remark}

Let us sketch the proof; the full proof can be found in the online appendix.

\xhdr{The ``good events".}
We establish four ``good events" each of which occurs with high probability.
\begin{description}
\item[(\event{1})] \emph{Exploration in Level 1:} Every first-level group collects at least $\Omega(T_1)$ samples of each arm.
\item[(\event{2})] \emph{Concentration in Level 1:} Within each first-level group, empirical mean rewards of each arm $a$ concentrate around $\mu_a$.
\item[(\event{3})] \emph{Anti-concentration in Level 1:} For each arm, some first-level subgroup collects data which makes this arm look much better than its actual mean and the other arm look much worse than its actual mean.
\item[(\event{4})] \emph{Concentration in prefix:}
The empirical mean reward of each arm $a$ concentrates around $\mu_a$ in any prefix of its pulls. (This ensures accurate reward estimates in exploitation.)
\end{description}

The analysis of these events applies Chernoff Bounds to a suitable version of ``reward tape" (see the definition of ``reward tape" in \Cref{sec:pfs-prelims}). For example, \event{2} considers a reward tape restricted to a given first-level group.

\xhdr{Case analysis.}
We now proceed to bound the regret conditioned on the four ``good events". W.l.o.g., assume $\mu_1 \geq \mu_2$. We break down the regret analysis into four cases, based on the magnitude the gap parameter $\Delta = \mu_1-\mu_2$. As a shorthand, denote
    $\conf{n} = \sqrt{\log(T)/n}$.
In words, this is a confidence term, up to constant factors, for $n$ independent random samples.

The simplest case is very small gap, trivially yielding an upper bound on regret.

\begin{claim}[Negligible gap]
If
    $\Delta \leq 3\sqrt{2}\cdot  \conf{T_2}$
then
  $\reg(T)\leq O(T^{4/7} \log^{6/7}(T))$.
\end{claim}

Another simple case is when $\Delta$ is sufficiently large, so that
the data collected in any first-level group suffices to determine the best arm. The proof follows from \event{1} and \event{2}.

\begin{lemma}[Large gap]\label{3levelbigcase}
If
  $ \Delta \geq 4 \sum_{a\in\A}\; \conf{\fdpN[a]\cdot T_1}$
then all agents in the second and the third levels pull arm 1.
\end{lemma}

In the \emph{medium gap} case, the data collected in a given first-level group is no longer guaranteed to determine the best arm. However, agents in the third level see the history of \emph{all} first-level groups, which enables them to correctly identify the best arm.


\begin{lemma}[Medium gap]\label{3levelmedium}
  All agents pull arm 1 in the third level, when $\Delta$ satisfies
\[ \textstyle \Delta\in \left[
    4\,\sum_{a\in\A}\; \conf{\NG\cdot \fdpN[a]\cdot T_1},\quad
    4\,\sum_{a\in\A}\; \conf{\fdpN[a]\cdot T_1}
 \right]. \]
\end{lemma}

Finally, the \emph{small gap} case, when  $\Delta$ is between
$\tilde\Omega(\sqrt{1/T_2})$ and $\tilde O(\sqrt{1/(\NG\, T_1)})$
is more challenging since even aggregating the data from all $\NG$
first-level groups is not sufficient for identifying the best arm.
We need to ensure that both arms continue to be explored in the second level.
To achieve this, we leverage \event{3}, which implies
that each arm $a$ has a first-level group $s_a$ where it gets ``lucky", in the sense that its empirical mean reward is slightly higher than $\mu_a$, while the empirical mean reward of the other arm is slightly lower than its true mean. Since the
deviations are in the order of $\Omega(\sqrt{1/T_1})$, and Assumption~\ref{ass:embehave} guarantees the agents' reward estimates are also within $\Omega(\sqrt{1/T_1})$ of the empirical means, the sub-history
from this group $s_a$ ensures that all agents in the respective second-level group prefer arm $a$. Therefore, both arms are pulled at least $T_2$ times in the second level, which in turn gives the following guarantee:

\begin{lemma}[Small gap]\label{3levelsmallcase}
  All agents pull arm 1 in the third level, when 
\[ \textstyle
    \Delta\in \left( 3\sqrt{2}\cdot\conf{T_2},\quad
      4\,\sum_{a\in\A}\; \conf{\NG\cdot \fdpN[a]\cdot T_1}
\right).
\]
\end{lemma}

\paragraph{Wrapping up: proof of \Cref{thm:3level}. } In negligible
gap case, the stated regret bound holds regardless of what the policy does. In the large gap case, the regret only comes from the
first level, so it is upper-bounded by the total number of agents in this level, which is
    $\NG\cdot \GdT \cdot T_1 = O(T^{4/7} \log T)$.
In both intermediate cases, it suffices to bound the regret from the
first and second levels, so
\[ \textstyle
\reg(T) \leq (\NG\,T_1\cdot \GdT + \NG\, T_2)
\cdot 4\,\sum_{a\in\A}\; \conf{\fdpN[a]\cdot T_1}
= O(T^{4/7} \log^{6/7}(T)).
\]
Therefore, we obtain the stated regret bound in all cases.

\OMIT{
\begin{proof}
Wlog we assume $\mu_1 \geq \mu_2$ as the recommendation policy is symmetric to both arms. We do a case analysis based on $\mu_1-\mu_2$.

Before we start with the case analysis, we first define several clean events and show that the intersection of them happens with high probability.

\begin{itemize}

\item \textbf{Concentration of the number of arm $a$ pulls in the first level:}
\OMIT{By Lemma \ref{lem:greedy}, we know $\GdP \leq \fdpN  \leq \GdT$. For the $s$-th first-level group, define $W_1^{a,s}$ to be the event that the number of arm $a$ pulls in the $s$-th first-level group is between $\fdpN  T_1- \GdT \sqrt{T_1\log(T)}$ and $\fdpN  T_1 + \GdT \sqrt{T_1\log(T)}$. By Chernoff bound,}
For the $s$-th first-level group, define $W_1^{a,s}$ to be the event
that the number of arm $a$ pulls in the $s$-th first-level group is
between $\fdpN  T_1- \GdT \sqrt{T_1\log(T)}$ and
$\fdpN  T_1 + \GdT \sqrt{T_1\log(T)}$. By Lemma~\ref{lem:t1runs}
\[
\Pr[W_1^{a,s}] \geq 1-2\exp(-2\log(T)) \geq 1-2/T^2.
\]
Let $W_1$  be the intersection of all these events (i.e.
$W_1 = \bigcap_{a,s}W_1^{a,s}$). By union bound, we have
\[
\Pr[W_1] \geq 1- \frac{4S}{T^2}.
\]}

\OMIT{\item \textbf{Concentration of the empirical mean of arm $a$ pulls in the first level:}
For each first-level group and arm $a$, imagine there is a tape of enough arm $a$ pulls sampled before the recommendation policy starts and these samples are revealed one by one whenever agents in this group pull arm $a$. For the $s$-th first-level group and arm $a$, define $W_2^{s,a,t_1,t_2}$ to be the event that the mean of $t_1$-th to $t_2$-th pulls in the tape is at most $\sqrt{\frac{2\log(T)}{t_2-t_1+1}}$ away from $\mu_a$. By Chernoff bound,\swcomment{a bit confused about what $t_1$ and $t_2$ mean?}
\[
\Pr[W_2^{s,a,t_1,t_2}] \geq 1 - 2\exp(-4\log(T)) \geq 1- 2/T^4.
\]

Define $W_2$ to be the intersection of all these events (i.e. $W_2 = \bigcap_{a,s,t_1,t_2} W_2^{s,a,t_1,t_2}$). By union bound, we have
\[
\Pr[W_2] \geq 1- \frac{4S}{T^2}.
\]
}
\OMIT{\item \textbf{Concentration of the empirical mean of arm $a$ pulls in the first two levels:}

For all the groups in the first two levels and arm $a$, imagine there is a tape of enough arm $a$ pulls sampled before the recommendation policy starts and these samples are revealed one by one whenever agents in the first two levels pull arm $a$. Define $W_3^{a,t}$ to be the event that the mean of the first $t$ pulls in the tape is at most $\sqrt{\frac{2\log(T)}{t}}$ away from $\mu_a$. By Chernoff bound,
\[
\Pr[W_3^{a,t}] \geq 1 - 2\exp(-4\log(T)) \geq 1- 2/T^4.
\]
Define $W_3$ to be the intersection of all these events (i.e. $W_3 = \bigcap_{a,t} W_3^{a,t}$). By union bound, we have
\[
\Pr[W_3] \geq 1- \frac{4}{T^3}.
\]
}
\OMIT{\item \textbf{Anti-concentration of the empirical mean of arm $a$ pulls in the first level:}

Consider the tapes defined in the second bullet again. For the $s$-th first-level group and arm $a$, define $W_4^{s,a,high}$  to be the event that first $\fdpN  T_1$ pulls of arm $a$ in the corresponding tape has empirical mean at least $\mu_a + 1/\sqrt{\fdpN  T_1}$ and define  $W_4^{s,a,low}$  to be the event that first $\fdpN  T_1$ pulls of arm $a$ in the corresponding tape has empirical mean at most $\mu_a - 1/\sqrt{\fdpN  T_1}$. By Berry-Essen Theorem and $\mu_a \in [1/3,2/3]$, we have
\[
\Pr[W_4^{s,a,high}] \geq (1-\Phi(1/2)) - \frac{5}{\sqrt{\fdpN T_1}} > 1/4.
\]
The last inequality follows when $T$ is larger than some constant.
Similarly we also have
\[
\Pr[W_4^{s,a,low}] > 1/4.
\]
Since $W_4^{s,a,high}$ is independent with $W_4^{s,3-a,low}$, we have
\[
\Pr[W_4^{s,a,high} \cap W_4^{s,3-a,low}] =\Pr[W_4^{s,a,high}] \cdot  \Pr[W_4^{s,3-a,low}]>(1/4)^2 = 1/16.
\]
Now define $W_4$ as $\bigcap_a \bigcup_s (W_4^{s,a,high} \cap W_4^{s,3-a,low})$. Notice that $(W_4^{s,a,high} \cap W_4^{s,3-a,low})$ are independent across different $s$'s. By union bound, we have
\[
\Pr[W_4] \geq 1- 2(1-1/16)^S \geq 1 -2 /T.
\]
\end{itemize}

By union bound, the intersection of these clean events (i.e. $\bigcap_{i=1}^4 W_i$) happens with probability $1-O(1/T)$. When this intersection does not happen, since the probability is $O(1/T)$, it contributes $O(1/T) \cdot T = O(1)$ to the expected regret.

Now we assume the intersection of clean events happens and we summarize what these clean events imply.

\begin{itemize}
\item For the $s$-th first-level group and arm $a$, define $\bar{\mu}_a^{1,s}$ to be the empirical mean of arm $a$ pulls in this group. $W_1^{a,s}$, $W_2^{a,s,1,t}$ for $ = \fdpN  T_1- \GdT \sqrt{T_1\log(T)},...,\fdpN  T_1- \GdT \sqrt{T_1\log(T)}$ together imply that
\[
|\bar{\mu}_a^{1,s} - \mu_a| \leq \sqrt{\frac{2\log(T)}{\fdpN  T_1- \GdT \sqrt{T_1\log(T)}}} \leq \sqrt{\frac{4\log(T)}{\fdpN  T_1}}.
\]
The last inequality holds when $T$ is larger than some constant.
\item For each arm $a$, define $\bar{\mu}_a$ to be the empirical mean of arm $a$ pulls in the first two levels. $W_1^{a,s}$ for $s=1,...,S$ and $W_3^{a,t}$ for $t \geq  (\fdpN  T_1- \GdT \sqrt{T_1\log(T)})S$ together imply that
\[
|\bar{\mu}_a - \mu_a| \leq \sqrt{\frac{2\log(T)}{S\left(\fdpN  T_1- \GdT \sqrt{T_1\log(T)}\right)}} \leq \sqrt{\frac{4\log(T)}{S \fdpN  T_1}} .
\]
The last inequality holds when $T$ is larger than some constant.

If there are at least $T_2$ pulls of arm $a$ in the first two levels,
\[
|\bar{\mu}_a-\mu_a| \leq \sqrt{\frac{2\log(T)}{T_2}}.
\]

\item For each $a \in \{1,2\}$, $W_4$ implies that there exists $s_a$ such that $W_4^{s_a,a,high}$ and $W_4^{s_a,3-a,low}$ happen. $W_4^{s_a,a,high}$,  $W_1^{s_a,a}$, $W_2^{s_a,a,t, \fdpN T_1}$ for $t = \fdpN  T_1- \GdT \sqrt{T_1\log(T)}+1, ...,\fdpN T_1-1$ and $W_2^{s_a,a,\fdpN T_1,t}$ for $t= \fdpN T_1,...,\fdpN  T_1+ \GdT \sqrt{T_1\log(T)}$ together imply that
\begin{align*}
\bar{\mu}_a ^{1,s_a} &\geq \mu_a + \left(\fdpN T_1 \cdot \frac{1}{\sqrt{\fdpN T_1}} - \GdT \sqrt{T_1\log(T)} \cdot \sqrt{\frac{2\log(T)}{ \GdT \sqrt{T_1\log(T)}}} \right) \cdot \frac{1}{\fdpN  T_1+ \GdT \sqrt{T_1\log(T)}} \\
&> \mu_a + \frac{1}{4\sqrt{\fdpN T_1}}.
\end{align*}
The second last inequality holds when $T$ is larger than some constant.
Similarly, we also have
\[
\bar{\mu}_{3-a} ^{1,s_a} < \mu_{3-a}   - \frac{1}{4\sqrt{q_{3-a} T_1}}.
\]
\end{itemize}

Finally we proceed to the case analysis. We give upper bounds on the expected regret conditioned on the intersection of clean events.

\begin{itemize}

\item $\mu_1 - \mu_2 \geq 2\left(\sqrt{\frac{4\log(T)}{\fdpN[1]T_1}}
+ \sqrt{\frac{4\log(T)}{\fdpN[2]T_1}}\right)$. In this case, we want to show that agents in the second and the third levels all pull arm 1.

First consider the $s$-th second-level group. We know that
\[
\bar{\mu}_1^{1,s} - \bar{\mu}_2^{1,s} \geq \mu_1 -\mu_2 - \sqrt{\frac{4\log(T)}{\fdpN[1]T_1}} - \sqrt{\frac{4\log(T)}{\fdpN[2]T_1}} \geq  \sqrt{\frac{4\log(T)}{\fdpN[1]T_1}} + \sqrt{\frac{4\log(T)}{\fdpN[2]T_1}}.
\]
For any agent $t$ in the $s$-th second-level group, by Assumption \ref{ass:embehave}, we have
\begin{align*}
\hat{\mu}_1^t - \hat{\mu}_2^t &>\bar{\mu}_1^{1,s} - \bar{\mu}_2^{1,s} - \frac{c_m}{\sqrt{\fdpN[1]T_1/2}} - \frac{c_m}{\sqrt{\fdpN[2]T_1/2}}\\
&\geq  \sqrt{\frac{4\log(T)}{\fdpN[1]T_1}} + \sqrt{\frac{4\log(T)}{\fdpN[2]T_1}}- \frac{c_m}{\sqrt{\fdpN[1]T_1/2}} - \frac{c_m}{\sqrt{\fdpN[2]T_1/2}}\\
 &> 0.
\end{align*}
Therefore, we know agents in the $s$-th second-level group will all pull arm 1.

Now consider the agents in the third level group. Recall $\bar{\mu}_a$ is the empirical mean of arm $a$ in the history they see. We have
\[
\bar{\mu}_1 - \bar{\mu}_2 \geq \mu_1 -\mu_2 - \sqrt{\frac{4\log(T)}{S\fdpN[1]T_1}} - \sqrt{\frac{4\log(T)}{S\fdpN[2]T_1}} \geq  \sqrt{\frac{4\log(T)}{\fdpN[1]T_1}}
+ \sqrt{\frac{4\log(T)}{\fdpN[2]T_1}}.
\]
Similarly as above, by Assumption \ref{ass:embehave}, we know $\hat{\mu}_1^t - \hat{\mu}_2^t > 0$ for any agent $t$ in the third level. So we know agents in the third-level group will all pull arm 1. Therefore the expected regret is at most $S T_G T_1 = O(T^{4/7} \log^{6/7}(T))$.

\item $2\left(\sqrt{\frac{4\log(T)}{S\fdpN[1]T_1}}
+ \sqrt{\frac{4\log(T)}{S\fdpN[2]T_1}}\right) \leq \mu_1-\mu_2 < 2\left(\sqrt{\frac{4\log(T)}{\fdpN[1]T_1}}
+ \sqrt{\frac{4\log(T)}{\fdpN[2]T_1}}\right)$. In this case, we want to show agents in the third level all pull arm 1. Recall $\bar{\mu}_a$ is the empirical mean of arm $a$ in the first two levels. We have
\[
\bar{\mu}_1 - \bar{\mu}_2 \geq \mu_1 -\mu_2 - \sqrt{\frac{4\log(T)}{S\fdpN[1]T_1}} - \sqrt{\frac{4\log(T)}{S\fdpN[2]T_1}} \geq  \sqrt{\frac{4\log(T)}{S\fdpN[1]T_1}}
+ \sqrt{\frac{4\log(T)}{S\fdpN[2]T_1}}.
\]
For any agent $t$ in the third level, by Assumption \ref{ass:embehave}, we have
\begin{align*}
\hat{\mu}_1^t - \hat{\mu}_2^t &>\bar{\mu}_1 - \bar{\mu}_2 - \frac{c_m}{\sqrt{S\fdpN[1]T_1/2}} - \frac{c_m}{\sqrt{S\fdpN[2]T_1/2}}\\
&\geq  \sqrt{\frac{4\log(T)}{S\fdpN[1]T_1}} + \sqrt{\frac{4\log(T)}{S\fdpN[2]T_1}}- \frac{c_m}{\sqrt{S\fdpN[1]T_1/2}} - \frac{c_m}{\sqrt{S\fdpN[2]T_1/2}}\\
 &> 0.
\end{align*}
So we know agents in the third-level group will all pull arm 1. Therefore the expected regret is at most
\[
(S T_G T_1 + S T_2) \cdot 2\left(\sqrt{\frac{4\log(T)}{\fdpN[1]T_1}}
+ \sqrt{\frac{4\log(T)}{\fdpN[2]T_1}}\right) = O(T^{4/7} \log^{6/7}(T))
\]

\item $ 3\sqrt{\frac{2\log(T)}{T_2}} < \mu_1-\mu_2 < 2\left(\sqrt{\frac{4\log(T)}{S\fdpN[1]T_1}}
+ \sqrt{\frac{4\log(T)}{S\fdpN[2]T_1}}\right)$. In this case, we just need to make sure that agents in the third level all pull arm 1. To do so, we need both arms to be pulled at least $T_2$ rounds in the second level.

Now consider the $s_a$-th second-level group. We have
\begin{align*}
\bar{\mu}_a^{1,s_a} - \bar{\mu}_{3-a}^{1,s_a} &> \mu_a + \frac{1}{4\sqrt{\fdpN T_1}} -\mu_{3-a} +\frac{1}{4\sqrt{q_{3-a}T_1}} \\
&> \frac{1}{4\sqrt{\fdpN[1]T_1}}+ \frac{1}{4\sqrt{\fdpN[2]T_1}} - 2\left(\sqrt{\frac{4\log(T)}{S\fdpN[1]T_1}}
+ \sqrt{\frac{4\log(T)}{S\fdpN[2]T_1}}\right) \\
&\geq \frac{1}{8\sqrt{\fdpN[1]T_1}}+ \frac{1}{8\sqrt{\fdpN[2]T_1}}.
\end{align*}
For any agent $t$ in the $s_a$-th second-level group, by Assumption \ref{ass:embehave}, we have
\begin{align*}
\hat{\mu}_a^t - \hat{\mu}_{3-a}^t &>\bar{\mu}_a^{1,s_a} - \bar{\mu}_{3-a}^{1,s_a} - \frac{c_m}{\sqrt{\fdpN[1]T_1/2}} - \frac{c_m}{\sqrt{\fdpN[2]T_1/2}}\\
&\geq   \frac{1}{8\sqrt{\fdpN[1]T_1}}+ \frac{1}{8\sqrt{\fdpN[2]T_1}}- \frac{c_m}{\sqrt{\fdpN[1]T_1/2}} - \frac{c_m}{\sqrt{\fdpN[2]T_1/2}}\\
 &> 0.
\end{align*}
So we know agents in the $s_a$-th second-level group will all pull arm $a$. Therefore in the first two levels, both arms are pulled at least $T_2$ times. Now consider the third-level. We have
\[
\bar{\mu}_1 - \bar{\mu}_2  \geq \mu_1 -\mu_2 - 2\sqrt{\frac{2\log(T)}{T_2}} \geq \sqrt{\frac{2\log(T)}{T_2}}.
\]
Similarly as above, by Assumption \ref{ass:embehave}, we know $\hat{\mu}_1^t - \hat{\mu}_2^t > 0$ for any agent $t$ in the third level. So we know agents in the third-level group will all pull arm 1.

Therefore the expected regret is at most
\[
(S T_G T_1 + S T_2) \cdot 2\left(\sqrt{\frac{4\log(T)}{S\fdpN[1]T_1}}
+ \sqrt{\frac{4\log(T)}{S\fdpN[2]T_1}}\right) \leq O(T^{4/7} \log^{6/7}(T))
\]

\item $\mu_1 - \mu_2 \leq 3\sqrt{\frac{2\log(T)}{T_2}}$. This is the easy case. Even always pulling the sub-optimal arm (i.e. arm 2) gives regret at most $T \cdot (\mu_1-\mu_2) = O(T^{4/7} \log^{6/7}(T))$.
\end{itemize}
\end{proof}}


\section{Optimal regret with a multi-level policy}
\label{sec:llevel}

We extend our three-level policy to a more adaptive multi-level policy in order to achieve the optimal regret rate of $\tildeO_K(\sqrt{T})$. This requires us to distinguish finer and finer gaps between the best and second-best arm.  A naive approach would be to recursively apply the 2-level structure, creating a tree of group aggregators, each level responsible for successively larger information sets.  This mimics the hierarchical information structure in many organizations, but it suffers large regret because the number of agents in focus groups grows exponentially.  Furthermore, each of these agents is forced to make decisions with access to a vanishingly-small amount of history, which is undesirable in-and-of itself.  In this section, we describe a method of interlacing information to reuse it without suffering from introduced correlations.  This careful reuse of information is the third and final step in our journey towards policies with optimal learning rates.

On a very high level, our multi-level policy implement the limited-adaptivity framework for multi-armed bandits \citep{Perchet2015BatchedBP}, defined is as follows. Suppose a bandit algorithm outputs a distribution $p_t$ over arms in each round $t$, and the arm $a_t$ is then drawn independently from $p_t$. This distribution can change only in a small number of rounds, called \emph{adaptivity rounds}, that need to be chosen by the algorithm in advance. Optimal regret rate requires at least $O(\log \log T)$ adaptivity rounds, where each ``level" $\ell\geq 2$ in our construction implements one adaptivity round. The limited-adaptivity bandit algorithm from \citet{Perchet2015BatchedBP} is much simpler compared to our construction below, as it can ensure the desired amount of exploration directly by choosing the appropriate alternatives.

We provide two results (for two different parameterizations of the same policy). The first result analyzes the $L$-level policy for an arbitrary $L\leq O(\log \log T)$, and achieves the root-$T$ regret rate with $O(\log \log T)$ levels.

\begin{theorem}
\label{thm:llevel-1}
There exists $L_{\max} = \Theta(\log\log T)$ such that
for each $L \in \{3, 4 \LDOTS L_{\max}\}$ there exists an order-based disclosure policy with $L$ levels and
regret
$$\reg(T) \leq O_K\rbr{ T^{\gamma} \cdot \polylog(T) },\quad
\text{where } \gamma = \frac{2^{L-1}}{2^L-1}.$$
In particular, we obtain regret
    $O_K(T^{1/2} \polylog(T))$
with $L= O(\log\log(T))$.
\end{theorem}

Our second policy achieves a gap-dependent regret
guarantee, as per \eqref{eq:model-OptRegret}. This policy has the same info-graph structure as the first
one in Theorem \ref{thm:llevel-1}, but requires a higher number of
levels $L = O(\log(T/\log\log(T)))$ and different group sizes. We will
bound its regret as a function of the gap parameter $\Delta$ even
though the construction of the policy does not depend on $\Delta$. In
particular, this regret bound outperforms the one in Theorem
\ref{thm:llevel-1} when $\Delta$ is much bigger than $T^{-1/2}$.  It
also has the desirable property that the policy does not withhold too
much information from agents---any agent $t$ observes a good fraction
of history in previous rounds.

\begin{theorem}
\label{thm:llevel-2}
There exists an order-based disclosure policy with
    $L = O(\log(T)/\log\log(T))$
levels such that for every bandit instance with gap parameter $\Delta$, the policy has regret
$$\reg(T) \leq O_K\rbr{ \min\rbr{ 1/\Delta, \; T^{1/2}} \cdot \polylog(T)}.$$
Under this policy,
each agent $t$ observes a subhistory of
size at least $\Omega( t/\polylog(T))$.
\end{theorem}

Note for constant number of arms, this result matches the optimal
regret rate (given in \Cref{eq:model-OptRegret}) for stochastic
bandits, up to logarithmic factors.

\begin{remark}
The multi-level policy can be applied to the first $T/\eta(T)$ agents only, for any fixed $\eta(T) = \polylog(T)$ (\ie, with reduced time horizon $T/\eta(T)$). Then the subsequent agents -- which comprise all but $1/\eta(T)$-fraction of the agents -- can observe the full history and enjoy instantaneous regret $\tildeO\rbr{T^{-1/2}}$. The regret bounds from both theorems carry over. This extension requires only minimal modifications to the analysis, which are omitted.
\end{remark}

Let us present our construction and the main ideas behind it; the full proofs are deferred to Section~\ref{sec:llevel-details}. We focus our explanations on the case of $K=2$ arms (but the construction itself applies to the general case).



A natural idea to extend the three-level policy is to insert more
levels as multiple ``check points", so the policy can incentivize the
agents to perform more adaptive exploration. In particular, each level will be responsible for some range of the gap parameter, collecting enough samples to rule out the bad arm if the gap parameter falls in this range. However, we need to
introduce two main modifications in the info-graph to accommodate some new challenges.


\OMIT{We will now present the main idea of extending from 3-level to
  $L$-level is that instead of using the second level as one
  ``check-point'', we use more levels to have multiple ``check
  points''.  Some new challenges appear in this process. Here we give
  an overview of the additional techniques we use to get our $L$-level
  results.}

\xhdr{Interlacing connections between levels.}
A tempting approach, described intuitively at the beginning of this section, generalizes the three-level policy to build an $L$-level
info-graph with the structure of a $\NG$-ary tree. Specifically, for every
$\ell\in \{2, \ldots , L\}$, each $\ell$-level group observes the
sub-history from a set of $\NG =\Theta(\log(T))$ groups in level $\ell-1$
such that these sets are mutually disjoint.
The
disjoint sub-histories observed by all the groups in level $\ell$ are
independent, and under the small gap regime (similar to
Lemma~\ref{3levelsmallcase}) it ensures that each arm has a
``lucky'' $\ell$-level group of agents that only pull this arm. This ``lucky''
property is crucial for ensuring that both arms will be explored in
level $\ell$.

However, in this construction, the first level will have $\NG^{L-1}$
groups, which introduces a multiplicative factor of $\NG^{\Omega(L)}$
to the regret rate. The exponential dependence in $L$ will heavily
limit the adaptivity of the policy, and prevents having the number of
levels for obtaining the result in \Cref{thm:llevel-2}. To overcome this, we will design an info-graph structure with $\NG^2 = \Theta(\log^2(T))$ groups at each level.%
\footnote{This exponential dependence on $L$ would not substantially affect \Cref{thm:llevel-1} (the one without the gap-dependent regret bound), and the mitigation -- the ``interlacing connection structure" described below -- may be unnecessary for this theorem. However, we present a unified construction for both theorems (and a mostly unified analysis), for the sake of simplicity. We note that \Cref{thm:llevel-1} does suffer from the ``boundary cases" issue described later in this section, and leverages the ``amplifying groups" aspect of our construction.}

We will leverage the following key observation: in order to maintain
the ``lucky'' property, it suffices to have $\Theta(\log T)$ $\ell$-th
level groups that observe disjoint sub-histories that take place in
level $(\ell - 1)$. Moreover, as long as the group size in levels lower
than $(\ell-1)$ are substantially smaller than group size of level $\ell-1$,
the ``lucky'' property does not break even if different groups in
level $\ell$ observe overlapping sub-history from levels
$\{1, \ldots, \ell-2\}$.

\OMIT{In our 3-level policy, the second level has
  $\NG = \Theta(\log(T))$ groups. We do so to ensure that in the small
  gap case (Lemma \ref{3levelsmallcase}), with high probability, both
  arms are pulled enough times in the second level. The argument
  relies on the fact that agents in different second-level groups
  observe disjoint histories of the first level and the independent
  randomness in first level groups guarantees each arm $a$ to have a
  ``lucky'' second-level group such that agents in that group all pull
  arm $a$ with high probability.

  Simply generalizing this idea to an $L$-level policy would give us a
  $\NG$-ary tree like structure and the first level will have
  $\NG^{L-1}$ groups. It incurs an extra $\NG^{\Omega(L)}$ factor in
  the regret. If we use $\log\log(T)$ levels as in
  \Cref{thm:llevel-1}, this factor super poly-logarithmic. If we use
  $\log(T)/\log\log(T)$ levels as in \Cref{thm:llevel-2}, this
  factor goes up to polynomial in $T$.

  In order to avoid this undesirable factor, we design a slightly
  different connecting structure between levels in our $L$-level
  policy.  The key observation is that, for any level
  $\ell \in \{3,\ldots,L-1\}$, we do want agents in different $\ell$-th level
  groups to see disjoint histories of the $(\ell-1)$-th level to ensure
  that there is some ``lucky'' group for each arm whose mean is close
  enough to the best arm. However, it does not matter much if pulls in
  levels below $\ell-1$ get observed by agents in different $\ell$-th level
  groups because group sizes in lower levels are much smaller than
  group sizes of level $\ell-1$ and the independent randomness in level
  $\ell-1$ is sufficient.}

This motivates the following interlacing connection structure between
levels. There are $\NG^2$ groups in each level of the info-graph.
The groups in the $\ell$-th level are
labeled as $G_{\ell,u,v}$ for $u,v\in[\NG]$. For any $\ell \in \{2,\ldots,L\}$
and $u,v,w\in [\NG]$, agents in group $G_{\ell,u,v}$ see the history of
agents in group $G_{\ell-1,v,w}$ (and by transitivity all agents in
levels below $\ell-1$). See Figure \ref{fig:llevel-connecting} for a
visualization of simple case with $\NG = 2$). Two observations are
in order:
\begin{itemize}
\item[(i)] Consider level $(\ell - 1)$ and fix the last group index to be
  $v$, and consider the set of groups
  $\cG_{\ell-1, v}=\{G_{\ell-1,i,v}\mid i \in [\NG]\}$ (\eg $G_{\ell-1,1,1}$
  and $G_{\ell-1,2,1}$ circled in red in the Figure
  \ref{fig:llevel-connecting}). The agents in any group of
  $\cG_{\ell-1, v}$ observe the same sub-history. As a result, if the
  empirical mean of arm $a$ is sufficiently high in their shared
  sub-history, then all groups in $\cG_{\ell-1, v}$ will become ``lucky''
  for $a$.

\item[(ii)] Every agent in level $\ell$ observes the sub-history from
  $\NG$ $(\ell-1)$-th level groups, each of which belonging to a
  different set $\cG_{\ell-1, v}$. Thus, for each arm $a$, we just need
  one set of groups $\cG_{\ell-1, v}$ in level $\ell-1$ to be ``lucky'' for
  $a$ and then all agents in level $\ell$ will see sufficient arm $a$
  pulls.
 \OMIT{For each set of groups
    in level $\ell-1$, agents in level $\ell$ see the history of agents in
    one of these groups. Therefore, for each arm $a$, we just need one
    set of groups in level $\ell-1$ to get ``lucky'' and then all agents
    in level $\ell$ will see enough pulls of arm $a$.}
\end{itemize}

\begin{figure}[h]
\centering
\begin{tikzpicture}
 \foreach \x in {0,3,6,9}
 {
 \filldraw[fill=purple!20!white]
 (\x+0,7.5)--(\x+2,7.5)--(\x+2,8.5)--(\x+0,8.5)--cycle;
 \filldraw[fill=green!20!white]
 (\x+0,5)--(\x+2,5)--(\x+2,6)--(\x+0,6)--cycle;
 \filldraw[fill=blue!20!white]
 (\x+0,2.5)--(\x+2,2.5)--(\x+2,3.5)--(\x+0,3.5)--cycle;
 }
\foreach \y in {2.5,5,7.5}
{
  \node at (12.5,\y+0.5){$\cdots$};
}

\foreach \y in {3.5,6}
{
  \draw (1,\y) -- (1,\y+1.5);
  \draw (1,\y) -- (7,\y+1.5);

  \draw (4,\y) -- (1,\y+1.5);
  \draw (4,\y) -- (7,\y+1.5);

  \draw (7,\y) -- (4,\y+1.5);
  \draw (7,\y) -- (10,\y+1.5);

  \draw (10,\y) -- (4,\y+1.5);
  \draw (10,\y) -- (10,\y+1.5);
}

\foreach \u in {1,2}
{
	\foreach \v in {1,2}
	{
	\pgfmathsetmacro{\x}{((\u-1)*2+(\v-1))*3};
	\pgfmathsetmacro{\xa}{((\u-1))*3};
	\pgfmathsetmacro{\xb}{(2+(\u-1))*3};
   \node at(\x+1,8){$G_{\ell,\u,\v}$};
   \node at(\x+1,5.5){$G_{\ell-1,\u,\v}$};
   \node at(\x+1,3){$G_{\ell-2,\u,\v}$};
	}
}

  \node at(-1.2,3){\textbf{Level $\ell-2$}};
  \node at(-1.2,5.5){\textbf{Level $\ell-1$}};
  \node at(-1.2,8){\textbf{Level $\ell$}};
  \draw[->] (13.3,2)--(13.3,8.5);
  \node at(13.7,5)[ rotate=90]{Time};

 \draw [rounded corners=3mm, line width=1mm, red](-0.2,4.8)--(2.2,4.8)--(2.2,6.2)--(-0.2,6.2)--cycle;
  \draw [rounded corners=3mm, line width=1mm, red](5.8,4.8)--(8.2,4.8)--(8.2,6.2)--(5.8,6.2)--cycle;
\end{tikzpicture}
\caption{Connections between levels for the $L$-level policy, for $\NG=2$.}
\label{fig:llevel-connecting}
\end{figure}

\xhdr{Amplifying groups for boundary cases.} Recall in the
three-level policy, the medium gap case (Lemma \ref{3levelmedium})
corresponds to the case where the gap $\Delta$ is between
$\Omega\left(\sqrt{{1}/{T_1}}\right)$ and
$O\left(\sqrt{{\log(T)}/{T_1}}\right)$. This is a boundary case since
$\Delta$ is neither large enough to conclude that with high
probability agents in both the second level and the third level all
pull the best arm, nor small enough to conclude that both arms are
explored enough times in the second level (due to
anti-concentration). In this case, we need to ensure that agents in
the third level can eliminate the inferior arm. This issue is easily
resolved in the three-level policy since the agents in the third level
observe the entire first-level history, which consists of
$\Omega(T_1\log(T))$ pulls of each arm and provides sufficiently
accurate reward estimates to distinguish the two arms.

\OMIT{about the situation in which the sub-optimal arm does not get
  pulled enough times in the first two levels and some agents in the
  third level may pull the sub-optimal arm.  Such situation is
  naturally ruled out in our 3-level policy for the following
  reason. Agents in the third level observe the history of the entire
  first level while agents in the second level only observes the
  history of a single first-level group. For each arm, even if it is
  not explored enough times in the second level, its empirical mean
  concentrates closer to its actual mean in the history observed by a
  third-level agent than the one observed by a second-level
  agent. Therefore, although the gap in the medium gap case is just
  not large enough for agents in the second level to all pull the best
  arm, it is large enough for agents in the third level to all pull
  the best arm.}

In the $L$-level policy, such boundary cases occur for each
intermediate level $\ell\in \{2, \ldots, l-1\}$, but the issue mentioned
above does not get naturally resolved since the ratios between the
upper and lower bounds of $\Delta$ increase from $\Theta\left(\sqrt{\log(T)}\right)$
to $\Theta(\log(T))$, and it would require more observations from
level $(\ell-2)$ to distinguish two arms at level $\ell$. The reason for
this larger disparity is that, except the first level, our guarantee on
the number of pulls of each arm is no longer tight. For example,
as shown in Figure \ref{fig:llevel-connecting}, when we talk about
having enough arm $a$ pulls in the history observed by agents in
$G_{\ell,1,1}$, it could be that only agents in group $G_{\ell-1,1,1}$ are
pulling arm $a$ and it also could be that most agents in groups
$G_{\ell-1,1,1},G_{\ell-1,1,2},...,G_{\ell-1,1,\NG}$ are pulling arm
$a$. Therefore our estimate of the number of arm $a$ pulls can be off
by an $\NG=\Theta(\log(T))$ multiplicative factor. This ultimately
makes the boundary cases harder to deal with.

We resolve this problem by introducing an additional type of \emph{amplifying groups}, called $\Gamma$-groups. For each $\ell \in [L], u,v \in [\NG]$, we create a $\Gamma$-group $\Gamma_{\ell,u,v}$. Agents in $\Gamma_{\ell,u,v}$ observe the same history as the one observed by agents in $G_{\ell,u,v}$ and the number of agents in $\Gamma_{\ell,u,v}$ is $\Theta(\log(T))$ times the number of agents in $G_{\ell,u,v}$. The main difference between $G$-groups and $\Gamma$-groups is that the history of $\Gamma$-groups in level $\ell$ is not sent to agents in level $\ell+1$ but agents in higher levels. When we are in the boundary case in which we don't have good guarantees about the $(\ell+1)$-level agents' pulls, the new construction makes sure that agents in levels higher than $\ell+1$ get to see enough pulls of each arm and all pull the best arm.


\xhdr{Recap and parameters.}
Let us recap our construction. The agents are partitioned into $L$ mutually disjoint subsets called \emph{levels}. Each level $\ell\in[L]$ is further partitioned into mutually disjoint subsets called \emph{groups}. There are two types of groups: resp., \emph{$G$-groups} and \emph{$\Gamma$-groups}. Each level $\ell\geq 2$ consists of $\NG^2$ groups of each type, labeled $G_{\ell,u,v}$ and $\Gamma_{\ell,u,v}$, resp., where $u$ and $v$ range over $[\NG]$. Level $1$ consists of $\NG^2$ $G$-groups, likewise labeled $G_{1,u,v}$; each these groups in turn consists of  several (mutually disjoint) \ALGGs of $\GdT$ agents each, where $\GdT$ is from Section~\ref{sec:warmup}.

The info-graph is defined as follows. Agents in level $1$ only observe the history from their respective \ALGG. For each level $\ell\geq 2$, all agents in a given group $G_{\ell,u,v}$ observe (i) all the history in the first $\ell-2$ levels (both $G$-groups and $\Gamma$-groups) and (i) the history from group $G_{\ell-1,v,w}$ for all $w \in [\NG]$. The agents in group $\Gamma_{\ell,u,v}$ observe the same history as those in  $G_{\ell,u,v}$.

Aside from the global parameter $\NG = \Theta(\log T)$ mentioned above, the structure of each level $\ell\in [L]$ is determined by a parameter $T_\ell$. Specifically, we have $T_1$ full-disclosure paths in each Level-1 group. For each level $\ell\geq 2$, each $G$-group contains $T_\ell$ agents, and each $\Gamma$-group contains $(\sigma-1) T_\ell$ agents. Parameters $T_1 \LDOTS T_L$ are specified in \Cref{sec:llevel-details}, differently for the two theorems.

Numerically, we set $\NG = 2^{10}\log(T)$ throughout, and
    $ L_{\max} :=  \log\rbr{ \frac{\ln T}{\log \NG^4} }$
in Theorem~\ref{thm:llevel-1}.


\section{Robustness}
\label{sec:robust}
We provide several results to illustrate that our constructions are robust to small amounts of misspecification. All these results require only minor changes in the analysis, which are omitted. First, we observe that all parameters in all policies can be increased by a constant factor.%
\footnote{For the two-level policy, this is a special case of Remark~\ref{rem:2level-tweaks}. We present it here for consistency.}

\begin{proposition}[parameters]
\label{thm:robust-params}
All results hold even if all parameters increase by at most a constant factor: specifically, parameters $(\fdpL,T_1)$ for the two-level policy (Theorem~\ref{thm:2level}), parameters $(\fdpL,\sigma, T_1, T_2)$ for the three-level policy (Theorem~\ref{thm:3level}), and parameters $(\fdpL,\sigma; T_1 \LDOTS T_L)$ for the $L$-level policy (Theorems~ \ref{thm:llevel-1} and \ref{thm:llevel-2}).
\end{proposition}

Let us consider a more challenging scenario when the \emph{structure} of the communication network is altered, introducing correlation between parts of the constructions that are supposed to be isolated from one another. Recall from Example~\ref{ex:robust-global} that even a small amount of such correlation can be extremely damaging if it comes early in the game.
Nevertheless, we can tolerate some undesirable correlation when it is sufficiently ``local" or happens in later rounds. Informally, the existence of a local side channel between consumers does not necessarily break the regret guarantees. Families and friends can share recommendations and the reviews they've received if their social networks are sufficiently disjoint and information doesn't travel too far.

Formally, we define a generalization of the two-level policy in which the exploration level can be wired in an arbitrary way, as long as it contains sufficiently many paths that are sufficiently long and sufficiently isolated. Agents in these paths may observe some agents that lie outside of these paths, but not too many, and these outside agents may not be shared among the paths. We need a definition: for a given subset $S$ of rounds, the \emph{span} of $S$ is the union of $S$ and all rounds $s$ that are observable in some round $t\in S$ (\ie rounds $s\leq t$ such that $s$ and $t$ are connected in the info-graph). We use quantity $\fdpL$ from Lemma~\ref{lem:greedy}.

\begin{proposition}[Robustness of the two-level policy]
\label{thm:robust-2level}
Fix some $N<T$. Consider an order-based disclosure policy such that each agent $t>N$ sees the full history: $S_t = [t-1]$. Suppose the info-graph on the first $N$ agents contains $M$ paths of length $\fdpL$ such that their spans are mutually disjoint and contain at most
    $2\cdot \fdpL$
rounds each. Then
\[ \reg(T) \leq \tilde{O}_K \rbr{ N + T/\sqrt{M}} . \]
In particular, we obtain
    $\reg(T) \leq \tilde{O}_K \rbr{ T^{2/3} }$
when $M = N = O(T^{2/3})$.
\end{proposition}

It is essential to bound the span size of the paths. Recall from Example~\ref{ex:robust-local} that too many ``leaf agents" observed by everyone in a given full-disclosure path would rule out the natural exploration in this path.

A similar but somewhat weaker result extends to multi-level policies.

\begin{proposition}[Undesirable correlations in Level 1]
\label{thm:robust-level1}
Consider the info-graph of either multi-layer policy (from  Theorem~\ref{thm:3level}, \ref{thm:llevel-1}, or \ref{thm:llevel-2}).
Suppose each full-disclosure path in Level 1 is replaced with subgraph
$H$ which contains at most  $2\cdot \fdpL$  rounds total, includes a path of length $\fdpL$, and is connected to the rest of the info-graph via $\max(H)$ only. Then the corresponding theorem still holds.
\end{proposition}

Moreover, we can handle some undesirable correlation outside of Level 1. As a proof of concept, we focus on the three-level disclosure policy, and allow each agent in Level $2$ to observe some  additional Level-$1$ agents. These agents can be chosen arbitrarily, \eg they could be the same for all Level-$2$ agents.

\begin{proposition}[undesirable correlations in Level $2$]
\label{thm:robust-global}
Consider the three-level policy from Theorem~\ref{thm:3level}. Add edges to the info-graph: connect each Level-$2$ agent to at most $O(\sqrt{T_1})$ arbitrarily chosen agents from Level-$1$, where $T_1$ is the parameter from Theorem~\ref{thm:3level}.  The resulting order-based  policy satisfies the guarantee in Theorem~\ref{thm:3level}.
\end{proposition}


\section{Numerical study}
\label{sec:expts}
\newcommand{\PoIE}{\term{PoIE}}
\newcommand{\LevOne}{\term{Level1}}
\newcommand{\FundProb}{\term{MaxMinExploration}}
\newcommand{\FDP}{full-disclosure path\xspace}
\newcommand{\FDPs}{full-disclosure paths\xspace}
\newcommand{\UCB}{\term{UCB}}
\newcommand{\LCB}{\term{LCB}}

We conduct a limited-scope numerical study to assess feasibility of our approach. We focus on the first level of our constructions (\ie parallel \FDPs, hence \LevOne), in the case of two arms. We numerically optimize its performance for various modeling choices, and discuss implications for the two-level construction. We compare the performance of our policies against bandit algorithms unrestricted by agents' incentives. Following \citet{Selke-PoIE-ec21}, the respective penalty in performance is called the \emph{price of incentivizing exploration} (\PoIE).

\xhdr{Foundations.}
We interpret \LevOne as a standalone algorithm which solves a fundamental exploration problem: choose each arm as much as possible within a given time horizon. Formally, given $T$ rounds, maximize $\min(N_1, N_2)$, where $N_a$ is the number of times arm $a\in[2]$ is chosen; call it \emph{\FundProb}.%
\footnote{\citet{Selke-PoIE-ec21} study a similar problem: minimize the number of rounds to collect $n$ samples of each arm. They work within the framing of BIC incentivized exploration, \ie not restricting to order-based disclosure policies, and theoretically characterize the optimal dependence on agents' beliefs and the number of arms.}
A primitive in our constructions, this problem is also significant on its own. Indeed, having collected the data, the platform could (i) exploit with respect to a different objective, misaligned with agents' rewards, and/or (ii) estimate some property other than the best arm. Moreover, the platform could consider different objectives/properties and they need not be known in advance.

For \LevOne, the objective $\min(N_1,N_2)$ simplifies as follows. Suppose \LevOne consists of $T_1$ \FDPs of length $P$ each. Then $\min(N_1,N_2)$ concentrates to $T_1\cdot \min\rbr{N_1^P,\, N_2^P}$ as $T_1$ increases, where $N_a^P$ is the expected number of times arm $a\in [2]$ is chosen within a single \FDP. (The difference can be upper-bounded by $O\rbr{P\sqrt{T_1}}$ via Chernoff Bounds.)

Thus, we have a simple comparison: \LevOne achieves $\min(N_1,N_2)\approx T_1\cdot \min\rbr{N_1^P,\, N_2^P}$, whereas an unrestricted bandit algorithm could sample each arm $T_1P/2$ times in the same time period.  With this in mind, we define \PoIE for \FundProb as the ratio of these two quantities:
\begin{align}
\PoIE := \min\rbr{N_1^P,\, N_2^P}\,/\, (P/2).
\end{align}
Note that $T_1$ cancels out, so that \PoIE is a property of a single \FDP.

This notion of \PoIE naturally connects to the two-level construction. Indeed, suppose an upper bound $\gamma\geq \PoIE$ is known. Then for large enough time horizon $T$ one could choose
    $T_1$
so as to achieve regret
    $R(T) \leq C\,T^{2/3}\cdot \rbr{\gamma \log T}^{1/3}  $,
for some (small) absolute constant $C>0$, as compared to the standard regret bound
    $R(T) \leq C\, T^{2/3}\cdot \rbr{\log T}^{1/3}  $
for the respective bandit algorithm of explore-then-commit.%
\footnote{We set $T_1\sim T^{2/3}\gamma^{1/3}/P$. The constant $C$ is the same in both cases, as it comes out of the same analysis.}
Thus, we have (only) $\gamma^{1/3}$ blow-up in the regret bound.

\xhdr{Experimental setup.}
We seek to characterize \PoIE across a wide range of problem instances. The problem instance here consists of a behavioral model (as per \Cref{sec:model-agents}) and a bandit instance $(\mu1,\mu_2)$. \PoIE can be numerically estimated for a particular problem instance and a particular $P$ by simulating the \FDP many times.%
\footnote{We used 100K iterates for \Cref{fig:expt-PoIE} and 20K iterates for \Cref{fig:expt-smallest-PoIE}.}
 The challenge is to reduce the ``search space" -- the number of problem instances considered.

We'll need some notation. Recall that $N_{t,a}$ and $\bar{\mu}_{t,a}$ denote the number of pulls and the empirical mean reward of arm $a$ in subhistory seen by agent $t$. Also, let $\UCB_{t,a}$ and $\LCB_{t,a}$ denote upper and lower confidence bounds allowed by \refeq{eq:ass:embehave} in our behavioral model, \ie
    $\bar{\mu}^t_a \pm \estC/\sqrt{N_{t,a}}$.

We focus on a particular ``shape" of problem instances. First, whenever a given arm $a$ is in the ``grey area", in the sense that $N_{t,a}\in [1, \estN)$, we posit that its reward estimate is not too small:
    $\bar{\mu}^t_a \geq \min\rbr{\nicefrac{1}{3}, \LCB_{t,a}}$.
Second, the mean rewards $\mu_1,\mu_2$ are in the range
    $[\nicefrac{1}{2}-\eps,\, \nicefrac{1}{2}+\eps]$
for some $\eps\in \sbr{0, \nicefrac{1}{2}}$. We call such problem instances \emph{\eps-well-formed}.

Within such problem instances, we identify the worst-case as far as \PoIE is concerned. Specifically, we skew everything for arm $1$ and against arm $2$:
\begin{OneLiners}
\item $\bar{\mu}^t_1=\UCB_{t,1}$ if $N_{t,1}\geq \estN$ and $\bar{\mu}^t_a=1$ otherwise;
\item $\bar{\mu}^t_2=\LCB_{t,1}$ if $N_{t,2}\geq \estN$ and $\bar{\mu}^t_a=\min\rbr{\nicefrac{1}{3}, \LCB_{t,a}}$ otherwise;
\item if $\bar{\mu}^t_1 = \bar{\mu}^t_2$, then arm $1$ is chosen;
\item $\mu_1 = 1+\eps$ and $\mu_2 = 1-\eps$.
\end{OneLiners}
Call this a \emph{canonical} problem instance with gap $2\eps$ (and parameters $\estN$ and $\estC$). It is indeed a worst-case for \PoIE, as per the following lemma (proved in \Cref{app:expts-proof}).

\begin{lemma}\label{lem:expts-canon}
Fix parameters $\estN,\estC$, path length $P$, and  $\eps\in \sbr{0, \nicefrac{1}{2}}$. Let $\mI$ be any $\eps$-well-formed problem instance with these parameters such that $N_1^P(\mI)\geq N_2^P(\mI)$. Let $\mI'$ be the corresponding canonical problem instance with gap $2\eps$. Then
    $\PoIE(\mI)\leq \PoIE(\mI')$.
\end{lemma}

\begin{figure*}[t]
\includegraphics[width=.5\textwidth]{\PATH 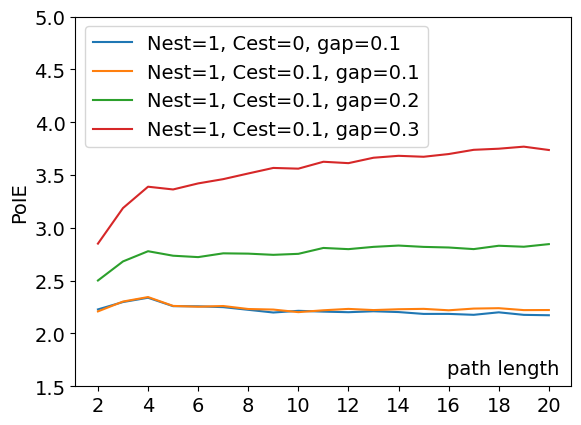}
\includegraphics[width=.5\textwidth]{\PATH 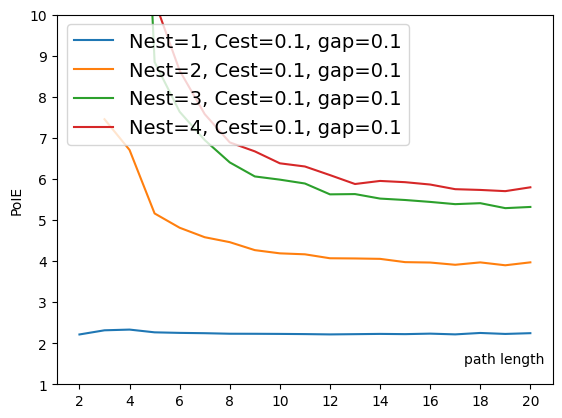}
\caption{PoIE as the path length grows: $\estN=1$ (left) and $\estN\in\cbr{1,2,3,4}$ (right).}
\label{fig:expt-PoIE}
\end{figure*}


\begin{wrapfigure}[13]{R}{0.5\textwidth}
  \begin{center}
\vspace{-6mm}
\includegraphics[width=.5\textwidth]{\PATH 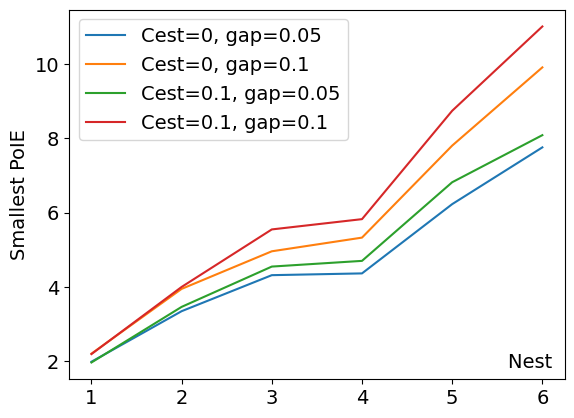}
  \end{center}
  \vspace{-5mm}
\caption{Smallest PoIE as $\estN$ grows.}
\label{fig:expt-smallest-PoIE}
\end{wrapfigure}

\xhdr{Experimental results.}
We focus on canonical problem instances defined above. We have several parameters: $\estN$, $\estC$, gap, and path length $P$. Note that $\estC=0$ corresponds to unbiased reward estimates, upon receiving enough at least $\estN$ samples. Recall that the paradigmatic case of our model, identified in Example~\ref{ex:paradigmatic}, posits $\estN=1$ and $\estC=0$.

In \Cref{fig:expt-PoIE}, we investigate how $\PoIE$ changes when $P$ grows, fixing everything else. The left-side plot focuses on $\estN=1$, whereas the right-side plot considers $\estN\in\cbr{1,2,3,4}$. We interpret $\text{gap}=0.1$ as ``already quite large" for multi-armed bandits, and $\estC=0.1$ as ``large enough" for our behavioral model. We observe that \PoIE is fairly stable as $P$ changes in the near-optimal region, and that the good/optimal values of $P$ are fairly small.

In \Cref{fig:expt-smallest-PoIE}, we optimize $P$ for a particular problem instance (obtaining the \emph{smallest} \PoIE), and investigate the change as $\estN$ grows. Each curve on this plot specifies the two remaining parameters, $\estC$ and $\text{gap}$. We plot 4 curves, corresponding to
       $\rbr{\estC,\, \text{gap}} \in \cbr{0,\, .1}\times\cbr{.05,\,.1}$.

From both figures, we conclude that $\estN$ is the crucial parameter for \PoIE, loosely corresponding to the strength of agents' beliefs. \PoIE is acceptable when $\estN$ is small, but tends to grow quickly as $\estN$  increases. The latter is consistent with the lower-bound result in \citet{Selke-PoIE-ec21} for BIC incentivized exploration. In that result, agents are endowed with a Bayesian prior based on $M$ data points, where $M$ is interpreted as the strength of beliefs.%
\footnote{Specifically, \citet{Selke-PoIE-ec21} consider conjugated Gaussian priors (resp., Beta-Bernoulli priors) obtained by conditioning the standard Gaussian prior (resp., a uniform prior) on $M$ data points.}
Then the smallest number of rounds needed to sample both arms is (at least) exponential in $M$.

\section{Conclusions}
\label{sec:conclusions}
We reformulate the problem of incentivized exploration as that of designing a fixed communication network for social learning. The new model substantially mitigates trust and rationality assumptions inherent in prior work on \OldProblem. We achieve optimally efficient social learning, in terms of how regret rate depends on the time horizon $T$. We do not restrict ourselves by the design choices adopted in the current recommendation platforms; instead, we hope to inform and  influence the designs in the future.

We start with a two-level communication network which is very intuitive and robust to misspecifications. The idea of splitting (some of) the early arrivals into many isolated ``focus groups" is plausibly practical. This construction implements the explore-then-commit paradigm from multi-armed bandits, and achieves vanishing regret. We obtain optimal regret rate via a more intricate, multi-level communication network. The conceptual challenge here is to make exploration optimally adaptive to past observation, despite the ``greedy" behavior of the agents. This requires intermediate ``group aggregators'' and information-sharing between groups, another plausibly actionable suggestion for recommendation platforms.

One could consider several well-motivated extensions of our model. However, they are not well-understood even in the BIC version, and arguably should first be resolved therein. To wit, one could
(i) allow reward-dependent biases in agents' reporting of the rewards.
(ii) allow long-lived agents that strive to optimize their long-term utility,
(iii) consider heterogenous agents, with public or private idiosyncratic signals, and
(iv) posit some unavoidable information leakage among the agents, \eg according to a pre-specified social network.
The first two extensions have not been studied even in the BIC version; the other two have been studied for \OldProblem, but are not yet well-understood.


\OMIT{ 
Another promising direction concerns further study of order-based disclosure policies, interpreted as a stylized model of social learning. Here, each agent faces one round of a multi-armed bandit problem, and the partial order specifies communication patterns between the agents. Rather than design order-based disclosure policies, let us characterize ones that admit (somewhat) efficient social learning, \eg ones that avoid herding.
} 



\bibliographystyle{plainnat}
\bibliography{bib-abbrv,bib-AGT,bib-bandits,bib-ML,bib-random,bib-competition,bib-slivkins,bib-nsi}

\newpage
\renewcommand{\contentsname}{ONLINE APPENDICES}
\tableofcontents

\crefalias{section}{appendix}
\crefalias{subsection}{appendix}

\begin{appendices}

\addtocontents{toc}{\setcounter{tocdepth}{1}} 

\newpage
\section{Prior work: greedy bandits and social learning}
\label{sec:discussion}

We present a more detailed discussion of some of the prior work, as alluded in Section~\ref{sec:related-work}.


\xhdr{Full disclosure and herding.}
The full-disclosure policy in \OldProblem reduces to the ``greedy" bandit algorithm which exploits in each round given the full history. Its herding effects are most lucidly summarized by focusing on the case of two arms and Bernoulli rewards.

For the Bayesian version, one has a Bayesian prior on the arm's mean rewards $(\mu_1,\mu_2)$, and the greedy algorithm chooses an arm with a largest posterior mean reward. Suppose arm $1$ is preferable according to the prior, \ie
    $\E[\mu_1-\mu_2]>0$.
Then the algorithm never tries arm $2$ with probability at least $\E[\mu_1-\mu_2]$.
This result holds for an arbitrary Bayesian prior, possibly correlated across arms. It implies very high regret (linear in $T$, the number of agents) under additional assumptions, \eg for independent priors with full support and for correlated priors with density bounded away from zero.

A similar but technically different result holds for a frequentist version, where the greedy algorithm chooses an arm with the highest empirical mean. Suppose $\mu_1,\mu_2$ are distinct and bounded away from $0$ and $1$. Assume a ``warm start" such that each arm is pulled $N_0$ times, for some $N_0<(\mu_1-\mu_2)^{-2}$. Then the best arm is never chosen with probability at least $\Omega(1/\sqrt{N_0})$.
Moreover, this result extends to a flexible behavioral model similar to ours, with failure probability that deteriorates exponentially in $\estC$.


These results can be found in \citet{BSL-myopic23}, along with extensions to $K>2$ arms. 



\xhdr{Social learning.}
As mentioned in Section~\ref{sec:related-work}, our work can be interpreted as coordinating \emph{social learning}, as we design a network on which the social learning happens. A large literature on social learning studies agents that learn over time in a shared environment, with no principal to coordinate them. A prominent topic is the  presence or absence of herding phenomena.  Models vary across several dimensions, to wit:
how an agent acquires new information;
which information is transmitted to others;
what is the structure / properties of the communication network;
whether agents are long-lived or only act once;
whether they optimize rewards (via Bayesian rationality or frequentist behavior), or merely follow a rule-of-thumb.
Below we discuss several directions in social learning that are most relevant, and point our the major differences compared to our model.

First, ``sequential social learning" posits that agents observe private signals, but only the chosen actions of neighbors are observable in the future; see \cite{Golub-survey16} for a survey.
While the early work focuses on a complete communication network,%
\footnote{See, for example, \cite{Banerjee-qje92,Welch92,Bikhchandani-jpe92}, as well as a very general result in \citet{SmithSorensen-econometrica00}.}
further work considers the impact of the network topology.  In particular, \cite{acemoglu2011bayesian} and \cite{lobel2015information} show that in a perfect Bayesian equilibrium, learning happens asymptotically if neighborhoods are sufficiently expansive or independent, features echoed in our own constructions.
To contrast these models with ours, consider the social planner that has access to all agents' knowledge and chooses their actions. In sequential social learning, such a planner only needs to choose the best action given the previous agents' signals, \ie only needs to \emph{exploit}, whereas in our model it also needs to \emph{explore}. Also, herding in sequential social learning occurs due to restricted information flow (\ie private signals are not observable in the future), whereas in our model there are no private signals\footnote{This follows from the transitivity of subhistories.} and herding happens even with full disclosure.

Another line of work, starting from \citet{DeGroot74}, posits that agents use ``naive", mechanical rules-of-thumb, \eg form beliefs based on naive averaging of observations.%
\footnote{Our frequentist agents may behave similarly, albeit  with more justification (because the subhistories they observe are unbiased and transitive). The original paper of \citet{DeGroot74} and much subsequent work study agents that act repeatedly, updating their beliefs over rounds.} In particular, even naive agents learn asymptotically so long as the network is not too imbalanced \citep[\eg][]{golub2010naive}.
\cite{chandrasekhar2020testing}
show experimentally that such a behavioral model is a good predictor of human behavior in some scenarios.
\cite{dasaratha2019experiment} show similar results for sequential social learning.  Theoretically, \cite{dasaratha2020learning} use this model of naivety to study the question of how to design the social network in a sequential learning model so as to induce optimal learning rates. They observe that silo structures akin to our two-level policy improve learning rates.


Third, ``strategic experimentation'', starting from \citet{Bolton-econometrica99,Keller-econometrica05}, studies long-lived learning agents that observe both actions and rewards of one another; see \citet{Horner-survey16} for a survey. This is similar to our work in that the social planner also solves a version of multi-armed bandits. The main difference is that the agents are long-lived and engage in a complex repeated game where each player deploys an exploration policy but would prefer to free-ride on exploration by others. There are also important technical differences. Agents exactly optimize their Bayesian-expected utility (using the Markov Perfect Equilibrium as a solution concept), whereas we consider a flexible frequentist model. Also, the social-planner problem is a very \emph{different} bandit problem, with Bayesian prior, time-discounting, ``safe" arm that is completely known, and ``risky" arm that follows a stochastic process.

Finally, \citet{Lazer-ASQ07} consider a network of myopic learners that strive to solve the same bandit problem. However, their work differs from ours in many ways. It is motivated by distributed problem solving within an organization (rather than recommendation systems). The communication network is endogenous (rather than designed by the platform). The bandit problem features a deterministic rewards and a large, combinatorially structured action space (rather than randomized rewards and a relatively small, unstructured action space). The agents are long-lived, acting all at once, and retain only their best-observed solution (rather than full history). Several network types are studied via extensive numerical simulations.

\newpage

\section{Additional discussion: beyond our assumptions}
\label{sec:discussion-model}
Let us discuss several realistic concerns that are left beyond the scope of our model.

\xhdr{Engineering concerns.}
While practical implementation is not a primary concern in this paper, our policies can be implemented efficiently as the number of agents grows. The two- and three-level policies are easy to scale, being tree structures, (and also somewhat robust, as discussed in \Cref{sec:robust}). The full-disclosure paths at Level 1 and the aggregation at the top level
are straightforward to implement at scale. Implementing the middle layer (in the three-level policy) only requires keeping track of which full-disclosure paths from Level 1 map to which middle-layer group. All of this only takes $o(T)$ extra space and a constant amount of computation per agent to generate the subhistory when a new agent arrives.

The full-blown multi-level policy, while much more intricate, still admits a scalable centralized implementation. Indeed, the platform needs to store the high-level structure ---  the connections between the groups --- and the two-way mapping between groups and agents. Again, this only takes $o(T)$ extra space. This structure suffices to generate the subhistory when a new agent arrives, by traversing through the high-level structure.

We do not explicitly handle delayed feedback (which may be present) and decentralized implementation (which may be necessary). In practice, datapoints with delayed feedback could simply be omitted from the history. This may be a non-issue for Level 1: assuming agents are assigned to full-disclosure paths in a round-robin-like fashion, all feedback may get submitted by the time it is needed. And even if some feedback is delayed, Level-1 agents may see a long enough subhistory for our analysis of Level 1 to work out. If so, the two-level  and three-level policies should be fine. Moreover, the hierarchical structure of the 2- and 3-level policies makes them easily parallelizable across multiple servers. These considerations are less clear for the multi-level policy.

\xhdr{Inconsistent subhistories.}
Agents from different ``regions" of the info-graph will see different subhistories for the same arm, and these subhistories may be inconsistent with one another. This may be a concern in practice, especially if these agents could easily communicate. We offer several considerations to mitigate this concern:

\begin{enumerate}
\item Substantial inconsistencies in average scores will be rare (because of concentration) once the subhistories get sufficiently long. The number of samples in all our policies grows near-uniformly over time, without big jumps. Slight inconsistencies in average scores (and sample counts) are less likely to be noticed, and they may happen even for full-disclosure policies due to engineering details such as batching and communication delays between servers. If the average scores and sample counts are similar, human users are not likely to notice (or mind) the discrepancy in the actual datapoints.

\item Some of the most blatant inconsistency may be preventable: a recommendation system may figure out which people are most related to each other (\eg family members) and assign them to the same observation path in the info-graph, so that their subhistories are similar and consistent with one another. So, the related people can be assigned to the same full-disclosure path or the same ``group" in our constructions.


\item Some observable inconsistency is arguably a fair fame, since the (principle of) order-based disclosure policy is announced to the public. Each agent individually would plausibly be OK with an unbiased history that comprises a substantial fraction of the full history.



\end{enumerate}

While inconsistent subhistories are less desirable compared to the ideal scenario of a full-disclosure policy, such comparison is arguably unfair (since this policy is very bad from the learning perspective). A more fair comparison would be to the prior work on BIC incentivized exploration, where each agent is only presented with a direct recommendation, with no supporting information. And users would arguably prefer order-based policies due to their transparency.

\xhdr{Arriving arms.} In practice, arms may arrive over time, \eg a new restaurant might open. While this issue is beyond our scope, we note that one can restart the whole mechanism whenever a new arm arrives. The resulting mechanism retains the incentives properties of order-based policies, and only increases regret by a factor of $O(\sqrt{K})$, where $K=\#\text{arms}$.%
\footnote{This easily follows from our $\sqrt{T}$-regret result: letting $T_i$ be the time between the $i$-th and $(i+1)$-th arm arrival, we obtain regret $\sum_{i=1}^K O(f(i) \sqrt{T_i}) \leq f(K) O(\sqrt{KT}) $,
where $f(K)$ is the (exponential) dependence on $K$ in our regret bound. Regret bounds for 2- and 3-level policies can be extended similarly.}
This is reasonable for a ``constant" $K$ (and not a major blow-up as $K$ grows since the dependence on $K$ is already exponential).
While this ``restarting trick" may be undesirable in practice, it is quite common in learning theory.%
\footnote{\Eg a well-known ``doubling trick" starts with an algorithm that has a regret bound for a known time horizon, restarts this algorithm at exponential times $t$ (\eg $t=2^i$, $i\in\N$), and achieves a regret bound that holds for every $t$.}
On the other hand, incorporating new arms without restarts and with roughly uniform accumulation of data over time may require some new ideas.

\section{Additional preliminaries}
\label{sec:pfs-prelims}
Let us spell out additional preliminaries needed for the proofs in the appendices.

We use the standard concentration and anti-concentration inequalities: respectively, Chernoff Bounds and Berry-Esseen Theorem. The former states that
    $\bar{X} = \frac{1}{n}\sum_{i=1}^n X_i$,
the average of $n$ independent random variables $X_1 \LDOTS X_n$, converges to its expectation quickly. The latter states that the CDF of an appropriately scaled average $\bar{X}$ converges to the CDF of the standard normal distribution pointwise. In particular, the average strays far enough from its expectation with some guaranteed probability. The theorem statements are as follows:

\begin{theorem}\label{thm:app-concentration}
Let $X_1 \LDOTS X_n$  be independent random variables, and $\bar{X} = \frac{1}{n}\sum_{i=1}^n X_i$.

\begin{OneLiners}
\item[(a)] \emph{(Chernoff Bounds)} 
Assume $X_i \in [0,1]$ for all $i$. Then
\[
\Pr[ |\bar{X} - \E[\bar{X}]| > \varepsilon] \leq 2\exp(-2n\varepsilon^2).
\]

\item[(b)] \emph{(Berry-Esseen Theorem)} 
Assume $X_1 \LDOTS X_n$ are i.i.d., with
    $\sigma^2 := \E\rbr{ (X_1 - \E[X_1])^2}$
and
    $\rho := \E\rbr{ |X_1 - \E[X_1]|^3} <\infty$.
Let $F_n$, $\Phi$ be the cumulative distribution functions of, resp.,
$(\bar{X} - \E[\bar{X}]) \frac{\sqrt{n}}{\sigma}$ and the standard normal distribution. Then
$|F_n(x) - \Phi(x) | \leq \rho/\rbr{2\sigma^3\sqrt{n}}$
for all $x\in\R$.

\end{OneLiners}
\end{theorem}

We use the notion of \emph{reward tape} to simplify the application of (anti-)concentration inequalities. This is a $K\times T$ random matrix with rows and columns corresponding to arms and rounds, respectively. For each arm $a$ and round $t$, the value in cell $(a,t)$ is drawn independently from Bernoulli distribution $\mD_a$. W.l.o.g., rewards in our model are defined by the rewards tape: namely, the reward for the $j$-th pull of arm $a$ is taken from the $(a,j)$-th entry of the reward matrix.

\section{The two-level policy (Theorem~\ref{thm:2level})}
\label{sec:pfs-2level}

We state two lemmas (which are also used to analyze the three- and multi-level policies). First, full-disclosure paths sample each arm with constant probability.

\begin{lemma}\label{lem:greedy}
There exist numbers $\fdpL>0$ and $\fdpP>0$ that depend only on $K$, the number of arms, with the following property. Consider an arbitrary disclosure policy, and let $S\subset [T]$ be a full-disclosure path in its info-graph, of length $|S|\geq \fdpL$. Under Assumption \ref{ass:embehave}, with probability at least $\fdpP$, it holds that subhistory $\SubH{S}$ contains at least one sample of each arm $a$.
\end{lemma}

\begin{proof}
  Fix any arm $a$. Let $\GdT = (K-1) \cdot \estN + 1$ and
  $\GdP = (1/3)^{\GdT}$. We will condition on the event that
    all the realized rewards in $\GdT$ rounds are 0, which occurs with
    probability at least $\GdP$ under Assumption~\ref{ass:embehave}.
  In this case, we want to show that arm $a$ is pulled at least
  once. We prove this by contradiction. Suppose arm $a$ is not pulled. By
  the pigeonhole principle, we know that there is some other arm $a'$
  that is pulled at least $\estN + 1$ rounds. Let $t$ be the round in
  which arm $a'$ is pulled exactly $\estN + 1$ times. By Assumption
  \ref{ass:embehave},
  $  \hat{\mu}_{a'}^t \leq 0 + \estC / \sqrt{\estN}  \leq \estC < 1/3$.
  On the other hand, we have
  $\hat{\mu}_a^t \geq 1/3 > \hat{\mu}_{a'}^t$. This contradicts 
  the fact that in round $t$, arm $a'\neq a$ is pulled.
\end{proof}

The second lemma concerns $\fdpN$, the expected number of samples of a given arm $a$ collected by a full-disclosure path of length $\fdpL$. It is a simple corollary of Theorem~\ref{thm:app-concentration}(a).

\begin{lemma}\label{lem:t1runs}
Suppose the info-graph contains $T_1$ full-disclosure paths of $\fdpL$ rounds each. Let $N_a$ be the number of samples of arm $a$ collected by all paths. Then
  \[
   \Pr\sbr{ | N_a - \fdpN T_1| \leq \fdpL\cdot \sqrt{T_1 \log(2K/\delta) / 2}
    \quad\text{for all arms $a\in \A$}} \geq 1-\delta.
  \]
\end{lemma}

We are now ready to prove Theorem~\ref{thm:2level}. We will set $T_1$ later in the proof, depending on whether the gap parameter $\Delta$ is known. For now, we just need to know we will
  make $T_1 \geq \frac{4(\GdT)^2}{(\GdP)^2}\log(T)$. Since this policy is
  agnostic to the indices of the arms, we assume w.l.o.g. that arm 1
  has the highest mean.

  The first $T_1 \cdot \GdT$ rounds will get total regret at most
  $T_1 \cdot \GdT$.  We focus on bounding the regret from the second
  level of $T - T_1 \cdot \GdT$ rounds. We consider the following two
   events. We will first bound the probability that both of them
  happen and then we will show that they together imply upper bounds
  on $|\hat{\mu}^t_a - \mu_a|$'s for any agent $t$ in the second
  level. Recall $\hat{\mu}^t_a$ is the estimated mean of arm $a$ by
  agent $t$ and agent $t$ picks the arm with the highest
  $\hat{\mu}^t_a$.

  \OMIT{\paragraph{Concentration of the number of arm $a$ pulls in the first
    level.}
By Lemma \ref{lem:greedy}, we know $\GdP \leq \fdpN \leq \GdT$.}
  Define $W_1^a$ to be the event that the number of arm $a$ pulls in
  the first level is at least $\fdpN T_1- \GdT \sqrt{T_1\log(T)}$.
  As long as we set
    $T_1 \geq 4\rbr{\GdT/\GdP}^2 \log(T)$,
    this implies that the number of arm $a$ pulls is then at least
    $\fdpN T_1/2$.
\OMIT{  By Chernoff bound,
  \[
    \Pr[W_1^a] \geq 1-\exp(-2\log(T)) \geq 1-1/T^2.
  \]
}
Let $W_1 = \bigcap_{a}W_1^a$ be the intersection of all these events. By Lemma~\ref{lem:t1runs}, we have
$\Pr[W_1] \geq 1- \frac{K}{T^2} \geq 1 - \tfrac{1}{T}.
$

\OMIT{\paragraph{Concentration of the empirical mean of arm $a$ pulls
    in the first level.}}  Next, we show that the empirical mean of
each arm $a$ is close to the true mean. To facilitate our reasoning,
let us imagine there is a tape of length $T$ for each arm $a$, with
each cell containing an independent draw of the realized reward from
the distribution $\cD_a$. Then for each arm $a$ and any $\z\in [T]$, we
can think of the sequence of the first $\z$ realized rewards of $a$
coming from the prefix of $\z$ cells in its reward tape. Define
$W^{a,\z}_2$ to be the event that the empirical mean of the first $\z$
realized rewards in the tape of arm $a$ is at most
$\sqrt{\frac{2\log(T)}{\z}}$ away from $\mu_a$. Define $W_2$ to be the
intersection of these events (i.e.  $\bigcap_{a,\z\in[T]} W^{a,\z}_2$).  By
Chernoff bound,
\[
\Pr[W^{a,\z}_2] \geq 1 - 2\exp(-4\log(T)) \geq 1-2/T^4.
\]
By union bound,
$
\Pr[W_2] \geq 1 - KT \cdot \frac{2}{T^4} \geq 1 - \frac{2}{T}.
$

By union bound, we have $\Pr[W_1 \cap W_2] \geq 1 - 3/T$. For the
remainder of the analysis, we will condition on the event
$W_1 \cap W_2$. Denote $\Lambda = \fdpN T_1/2$ for brevity.

Fix arm $a$ and agent $t$ in the second level. By events $W_1$ and $W_2$, we have
    $|\bar{\mu}^t_a - \mu_a| \leq \sqrt{2\log(T)/\Lambda}$.
By $W_1$ and Assumption \ref{ass:embehave}, we have
    $|\bar{\mu}^t_a - \hat{\mu}^t_a| \leq \estC/\sqrt{\Lambda}$.
Therefore,
\[
|\hat{\mu}^t_a - \mu_a|\leq \sqrt{2\log(T)/\Lambda}+\estC/\sqrt{\Lambda}
    \leq 3 \Phi,
    \quad\Phi :=\sqrt{\log(T)\,/ \, \rbr{\GdP T_1}}.
\]
So the second-level agents will pick an arm $a$ whose mean reward is at most $6\Phi$ away from the best arm.
To sum up, the total regret is at most
$T_1 \cdot \GdT + T \cdot (1-\Pr[W_1 \cap W_2]) + T \cdot  6 \Phi$.
By setting $T_1 = T^{2/3}\log(T)^{1/3}$, we get regret $O(T^{2/3}\log(T)^{1/3})$.

\section{Analysis of Example~\ref{ex:robust-global}}
\label{sec:robust-global}
We consider three events, denoted $\mE_1$, $\mE_2$, $\mE_3$. Event $\mE_1$ is that after the first $N_1=2$ rounds, arm 1 has empirical mean at most  $\mu' < \mu_2$ and arm 2 empirical mean at least $\mu_2$. (The proof can work for other constant $N_1$, too.) We pick $\mu'$ such that $\mu_2-\mu' =\Omega(1)$.  Event $\mE_2$ focuses on the next $N-N_1$ rounds. It asserts that arm $2$ is the only one chosen in these rounds, and the empirical mean in any prefix of these rounds is at least $\mu_2$. Event $\mE_3$ is that the last $T-N$ agents all choose arm 2.

We lower-bound $\Pr[\mE_1,\mE_2,\mE_3]$ by a positive constant by considering $\Pr[\mE_1]$, $\Pr[\mE_2\mid \mE_1]$ and $\Pr[\mE_3\mid  \mE_1, \mE_2]$. First, $\mE_1$ happens with a constant probability as arm 1 getting 0 in its first pull and arm 2 getting 1 in its first pull is a sub case of $\mE_1$.

Now we condition on $\mE_1$ happening. We show that $\mE_2$ happens with a positive-constant probability. We focus on the case when the first $N_2$ pulls of arm 2 in rounds $\{ N_1+1 \LDOTS N \}$
are all 1's for some large enough constant $N_2$ and then use Chernoff bound and union bound on the rest $N-N_1-N_2$ pulls.

\newcommand{\Etape}{\mE_{\mathtt{tape}}}

Now we condition on $\mE_1$ and $\mE_2$.
We consider a ``reward tape" generating rewards of arm 2, where the $t$-th ``cell"  in the tape corresponds to the reward of arm $2$ in round $t$ if this arm is chosen in this round. For each $t > N$, let $C_t$ be the subset of cells in the tape that correspond to rounds $S_t \cap  (N,T]$, where $S_t$ is the set of rounds observable by agent $t$. We can show that with very high probability, the empirical mean over $C_t$ is larger than $\mu'$ for all $t$. Let us focus on this event, call it $\Etape$. We show that under $\Etape$, each agents $t>N$ chooses arm $2$, using induction on $t$. This is because $C_t$, together with the history of the first $N$ rounds, is exactly the subhistory seen by agent $t$, if all agents in round $\{N+1,...,t-1\}$ pull arm 2.

\section{The three-level policy (Theorem~\ref{thm:3level})}
\label{sec:pfs-3level}

\subsubsection*{High-probability events (implied by Lemma~\ref{lem:t1runs})}


\begin{lemma}[Concentration of first-level number of pulls.]\label{3levelw1}
  Let $W_1$ be the event that for all groups $s\in [\NG]$ and arms
  $a\in \{1, 2\}$, the number of arm $a$ pulls in the $s$-th
  first-level group is within $\GdT \sqrt{T_1\log(T)}$ from $\fdpN  T_1$,
  where $\fdpN $ is the expected number of arm $a$ pulls in a \ALGG
  of length $\GdT$. Then $\Pr[W_1] \geq 1- \frac{4\NG}{T^2}$.
\end{lemma}

\begin{proof}
  For the $s$-th first-level group, define $W_1^{a,s}$ to be the event
  that the number of arm $a$ pulls in the $s$-th first-level group is
  between $\fdpN T_1- \GdT \sqrt{T_1\log(T)}$ and
  $\fdpN T_1 + \GdT \sqrt{T_1\log(T)}$. By Lemma~\ref{lem:t1runs},
    $\Pr[W_1^{a,s}] \geq 1-2\,e^{-2\log T} \geq 1-2/T^2$.
By union bound, the intersection
    $\bigcap_{a,s}W_1^{a,s}$,
has probability at least
    $1- \frac{4\NG}{T^2}$.
\end{proof}

To state the events, it will be useful to think of a
hypothetical reward tape $\cT^1_{s, a}$ of length $T$ for each
group $s$ and arm $a$, with each cell independently sampled from
$\cD_a$.  The tape encodes rewards as follows: the $j$-th time arm $a$
is chosen by the group $s$ in the first level, its reward is taken
from the $j$-th cell in this arm's tape. The following result
characterizes the concentration of the mean rewards among all
consecutive pulls among all such tapes, which follows from Chernoff
bound and union bound.

\begin{lemma}[Concentration of empirical means in the first level]\label{3levelw2}
  For any $\z_1, \z_2\in [T]$ such that $\z_1 < \z_2$, $s\in [\NG]$, and
  $a\in \{1,2\}$, let $W_2^{s,a,\z_1,\z_2}$ be the event that the mean
  among the cells indexed by $\z_1, (\z_1+1), \ldots, \z_2$ in the tape
  $\cT^1_{a, s}$ is at most $\sqrt{\frac{2\log(T)}{\z_2-\z_1+1}}$ away
  from $\mu_a$.  Let $W_2$ be the intersection of all these events
  (i.e.  $W_2 = \bigcap_{a,s,\z_1,\z_2} W_2^{s,a,\z_1,\z_2}$). Then
  $\Pr[W_2] \geq 1- \frac{4\NG}{T^2}$.
\end{lemma}

\begin{proof}
  By Chernoff bound,
$\Pr[W_2^{s,a,\z_1,\z_2}] \geq 1 - 2\,e^{-4\log T} \geq 1- 2/T^4$.
By union bound, we have
    $\Pr[W_2] \geq 1-  4\NG/T^2$.
\end{proof}

Our policy also relies on the anti-concentration of the empirical
means in the first round. We show that for each arm $a\in \{1, 2\}$,
there exists a group $s_a$ such that the empirical mean of $a$ is
slightly above $\mu_a$, while the other arm $(3 - a)$ has empirical
mean slightly below $\mu_{(3-a)}$. This event is crucial for inducing
agents in the second level to explore both arms when the their
mean rewards are indistinguishable after the first level.

\begin{lemma}[Co-occurence of high and low deviations in this first level]\label{3levelw4}
  For any group $s\in [\NG]$, any arm $a$, let $\tilde\mu_{a,s}$ be the
  empirical mean of the first $\fdpN  T_1$ cells in tape $\cT^1_{a, s}$.
  Let $W_3^{s,a,\text{high}}$ be the event
  $\tilde\mu_{a, s} \geq \mu_a + 1/\sqrt{\fdpN  T_1}$ and let
  $W_3^{s,a,\text{low}}$ be the event that
  $\tilde\mu_{a, s} \leq \mu_a - 1/\sqrt{\fdpN  T_1}$.  Let $W_3$ be the
  event that for every $a\in \{1, 2\}$, there exists a group
  $s_a\in [\NG]$ in the first level such that both $W_3^{s_a,a,\text{high}}$
  and $W_3^{s_a,3-a,\text{low}}$ occur. Then
    $\Pr[W_3]\geq 1 -2 /T$.
\end{lemma}

\begin{proof}
By Berry-Esseen Theorem (\Cref{thm:app-concentration}(b)) and the fact that 
  $\mu_a \in [1/3,2/3]$, for any arm $a$ it holds that
\[
\Pr\sbr{W_3^{s,a,high}} \geq (1-\Phi(\nicefrac12)) - 5/\sqrt{\fdpN T_1} > \nicefrac14.
\]
The last inequality follows when $T$ is larger than some constant.
Similarly we also have
    $\Pr\sbr{W_3^{s,a,low}} > \nicefrac14$.
Since $W_3^{s,a,high}$ is independent with $W_3^{s,3-a,low}$, we have
\[
\Pr\sbr{W_3^{s,a,high} \cap W_3^{s,3-a,low}}
    =\Pr\sbr{W_3^{s,a,high}} \cdot  \Pr\sbr{W_3^{s,3-a,low}}>(1/4)^2 = 1/16.
\]
Notice that events $W_3^{s,a,high} \cap W_3^{s,3-a,low}$ are independent
across different $s$. By union bound, we have
    $\Pr[W_3] \geq 1- 2(1-1/16)^\NG \geq 1 -2 /T$.
\end{proof}

Lastly, we will condition on the event that the empirical means of
both arms are concentrated around their true means in any prefix of
their pulls. This guarantees that the policy obtains an accurate
estimate of rewards for both arms after aggregating all the data in
the first two levels.


\begin{lemma}[Concentration of empirical means in the first two
  levels]\label{3levelw3}
  With probability at least $1 - \frac{4}{T^3}$, the following event
  $W_4$ holds: for all $a\in \{1, 2\}$ and $\z \in [N_{T, a}]$, the
  empirical means of the first $\z$
  arm $a$ pulls is at most
  $\sqrt{\frac{2\log(T)}{\z}}$ away from $\mu_a$, where $N_{T, a}$ is
  the total number of arm $a$ pulls by the end of $T$ rounds.
\end{lemma}

\begin{proof}
  For any arm $a$, let's imagine a hypothetical tape of length $T$,
  with each cell independently sampled from $\cD_a$. The tape encodes
  rewards of the first two levels as follows: the $j$-th time arm $a$
  is chosen in the first two levels, its reward is taken from the
  $j$-th cell in the tape. Define $W_4^{a,\z}$ to be the event that the
  mean of the first $t$ pulls in the tape is at most
  $\sqrt{\frac{2\log(T)}{\z}}$ away from $\mu_a$. By Chernoff bound,
\[
\Pr[W_4^{a,\z}] \geq 1 - 2\exp(-4\log(T)) \geq 1- 2/T^4.
\]
By union bound, the intersection
    $W_4$ of all these events
has
    $\Pr[W_4] \geq 1- 4/T^3$.
\end{proof}

Let $W = \bigcap_{i=1}^4 W_i$ be the intersection of all 4
events.  By union bound, $W$ occurs with probability $1-O(1/T)$. Note
that the regret conditioned on $W$ not occurring is at most
$O(1/T) \cdot T = O(1)$, so it suffices to bound the regret conditioned on $W$.

\subsubsection*{Case Analysis}

Now we assume the intersection $W$ of events $W_1,\cdots,W_4$ happens. We will
first provide some helper lemmas for our case analysis.

\begin{lemma}
  For the $s$-th first-level group and arm $a$, define
  $\bar{\mu}_a^{1,s}$ to be the empirical mean of arm $a$ pulls in
  this group. If $W$ holds, then
    $|\bar{\mu}_a^{1,s} - \mu_a| \leq \sqrt{4\log(T)\,/\,\rbr{\fdpN T_1}}$.
\end{lemma}

\begin{proof}
  The events $W_1$ and $W_2^{a,s,1,\z}$ for
  $\z = \fdpN T_1- \GdT \sqrt{T_1\log(T)},...,\fdpN T_1 + \GdT
  \sqrt{T_1\log(T)}$ together imply that
\[
|\bar{\mu}_a^{1,s} - \mu_a| \leq \sqrt{\frac{2\log(T)}{\fdpN T_1- \GdT \sqrt{T_1\log(T)}}} \leq \sqrt{\frac{4\log(T)}{\fdpN T_1}}.
\]
The last inequality holds when $T$ is larger than some constant.
\end{proof}

\begin{lemma}
  For each arm $a$, define $\bar{\mu}_a$ to be the empirical mean of
  arm $a$ pulls in the first two levels. If $W$ holds, then
    $|\bar{\mu}_a - \mu_a| \leq \sqrt{4\log(T)\,/\,\rbr{\NG \fdpN T_1}} .$
Furthermore, if there are at least $T_2$ pulls of arm $a$ in the first two levels,
\[
|\bar{\mu}_a-\mu_a| \leq \sqrt{2\log(T)/T_2}.
\]
\end{lemma}

\begin{proof}
The events $W_1$ and $W_4^{a,\z}$ for $\z \geq  (\fdpN T_1- \GdT \sqrt{T_1\log(T)})\NG$ jointly imply
  \[
    |\bar{\mu}_a - \mu_a| \leq \sqrt{\frac{2\log(T)}{\NG\left(\fdpN T_1- \GdT \sqrt{T_1\log(T)}\right)}} \leq \sqrt{\frac{4\log(T)}{\NG \fdpN T_1}} .
\]
The last inequality holds when $T$ is larger than some constant.
\end{proof}

\begin{lemma}\label{lem:luck}
  For the $s$-th first-level group and arm $a$, define
  $\bar{\mu}_a^{1,s}$ to be the empirical mean of arm $a$ pulls in
  this group. For each $a \in \{1,2\}$, there exists a group $s_a$
  such that
\[
\bar{\mu}_a ^{1,s_a} > \mu_a + 1/4\sqrt{\fdpN T_1} \quad \mbox{and} \quad
\bar{\mu}_{3-a} ^{1,s_a} < \mu_{3-a}   - 1/4\sqrt{\fdpN[3-a] T_1}.
\]
\end{lemma}

\begin{proof}
Denote $\Psi = \sqrt{T_1\log T}$ for brevity.
  For each $a \in \{1,2\}$, $W_3$ implies that there exists $s_a$ such
  that both $W_3^{s_a,a,high}$ and $W_3^{s_a,3-a,low}$ happen.  The
  events $W_3^{s_a,a,high}$, $W_1$, $W_2^{s_a,a,\z, \fdpN T_1}$
  for $\z = \fdpN T_1- \GdT \Psi+1, ...,\fdpN T_1-1$ and
  $W_2^{s_a,a,\fdpN T_1,\z}$ for
  $\z= \fdpN T_1,...,\fdpN T_1+ \GdT \Psi$ together imply that
\begin{align*}
\bar{\mu}_a ^{1,s_a} &\geq \mu_a + \left(\fdpN T_1 \cdot \frac{1}{\sqrt{\fdpN T_1}} - \GdT \Psi \cdot \sqrt{\frac{2\log(T)}{ \GdT \Psi}} \right) \cdot \frac{1}{\fdpN T_1+ \GdT \Psi}
> \mu_a + \frac{1}{4\sqrt{\fdpN T_1}},
\end{align*}
when $T$ is larger than some constant.
Similarly, we
$\bar{\mu}_{3-a} ^{1,s_a} < \mu_{3-a}   - \tfrac14/\sqrt{\fdpN[3-a] T_1}.$
\end{proof}

Now we proceed to the case analysis.

\begin{proof}[Proof of Lemma~\ref{3levelbigcase} (Large gap case)]
For brevity, denote
    $\Lambda := \rbr{\fdpN[1]T_1}^{-1/2} + \rbr{\fdpN[2]T_1}^{-1/2}$.
Observe that for any group $s$ in the first level, the empirical
  means satisfy
\[
\bar{\mu}_1^{1,s} - \bar{\mu}_2^{1,s}
    \geq \mu_1 -\mu_2 - \Lambda\sqrt{4\log T}
    \geq \Lambda\sqrt{4\log T}.
\]

For any agent $t$ in the $s$-th second-level group, by Assumption \ref{ass:embehave}, we have
\begin{align*}
\hat{\mu}_1^t - \hat{\mu}_2^t
>\bar{\mu}_1^{1,s} - \bar{\mu}_2^{1,s} - \Lambda\estC/\sqrt{2}
\geq
    \Lambda\rbr{\sqrt{4\log T}-\estC/\sqrt{2}}
    > 0.
\end{align*}
Therefore, we know agents in the $s$-th second-level group will all pull arm 1.

Now consider the agents in the third level group. Recall $\bar{\mu}_a$
is the empirical mean of arm $a$ in the history they see. We have
\[
\bar{\mu}_1 - \bar{\mu}_2
    \geq \mu_1 -\mu_2
        -\Lambda\sqrt{4\log(T)/\NG}
    \geq  \Lambda\sqrt{4\log(T)}.[
\]
Similarly as above, by Assumption \ref{ass:embehave}, we know
$\hat{\mu}_1^t - \hat{\mu}_2^t > 0$ for any agent $t$ in the third
level. Therefore, the agents in the third-level group will all pull
arm 1.  \OMIT{Therefore the expected regret is at most
  $\NG \GdT T_1 = O(T^{4/7} \log^{6/7}(T))$.}
\end{proof}

\begin{proof}[Proof of Lemma~\ref{3levelmedium} (Medium gap case)]
For brevity, denote
    $\Lambda := \rbr{\NG\fdpN[1]T_1}^{-1/2} + \rbr{\NG\fdpN[2]T_1}^{-1/2}$.

  Recall $\bar{\mu}_a$ is
  the empirical mean of arm $a$ in the first two levels. We have
\[
\bar{\mu}_1 - \bar{\mu}_2 \geq \mu_1 -\mu_2
    -\Lambda\sqrt{4\log T}
    \geq \Lambda\sqrt{4\log T}.
\]
For any agent $t$ in the third level, by Assumption \ref{ass:embehave}, we have
\begin{align*}
\hat{\mu}_1^t - \hat{\mu}_2^t &>\bar{\mu}_1 - \bar{\mu}_2
    -\Lambda\estC/\sqrt{2}
\geq -\Lambda\rbr{\sqrt{4\log T}-\estC/\sqrt{2}}
 > 0.
\end{align*}
So we know agents in the third-level group will all pull arm 1. \OMIT{Therefore the expected regret is at most
\[
(\NG \GdT T_1 + \NG T_2) \cdot 2\left(\sqrt{\frac{4\log(T)}{\fdpN[1]T_1}}
+ \sqrt{\frac{4\log(T)}{\fdpN[2]T_1}}\right) = O(T^{4/7} \log^{6/7}(T))
\]
}
\end{proof}

\begin{proof}[Proof of Lemma~\ref{3levelsmallcase} (Small gap case)]
We need both arms to be pulled at least $T_2$ rounds
  in the second level. For every arm $a$, consider the $s_a$-th
  second-level group, with $s_a$ given by Lemma~\ref{lem:luck}.
Denote
    $\Lambda := \rbr{\fdpN[1]T_1}^{-1/2} + \rbr{\fdpN[2]T_1}^{-1/2}$.
  We have
\begin{align*}
\bar{\mu}_a^{1,s_a} - \bar{\mu}_{3-a}^{1,s_a}
    &> \mu_a + \frac14/\sqrt{\fdpN T_1} -\mu_{3-a} +\frac14/\sqrt{\fdpN[3-a]T_1}
> \Lambda/4-4\Lambda\sqrt{4\log(T)/\NG}
\geq \Lambda/8.
\end{align*}
For any agent $t$ in the $s_a$-th second-level group, by Assumption \ref{ass:embehave}, we have
\begin{align*}
\hat{\mu}_a^t - \hat{\mu}_{3-a}^t
    &>\bar{\mu}_a^{1,s_a} - \bar{\mu}_{3-a}^{1,s_a}
        -\Lambda\estC/\sqrt{2}
\geq   \Lambda\rbr{\nicefrac{1}{8}-\estC/\sqrt{2}}
 > 0.
\end{align*}
So, we know agents in the $s_a$-th second-level group will all pull arm $a$. Therefore in the first two levels, both arms are pulled at least $T_2$ times. Now consider the third level. We have
\[
\bar{\mu}_1 - \bar{\mu}_2  \geq \mu_1 -\mu_2 - 2\sqrt{2\log(T)/T_2} \geq \sqrt{2\log(T)/T_2}.
\]
Similarly as above, by Assumption \ref{ass:embehave}, we know $\hat{\mu}_1^t - \hat{\mu}_2^t > 0$ for any agent $t$ in the third level. So we know agents in the third-level group will all pull arm 1.\OMIT{
Therefore the expected regret is at most
\[
(\NG \GdT T_1 + \NG T_2) \cdot 2\left(\sqrt{\frac{4\log(T)}{\NG\fdpN[1]T_1}}
+ \sqrt{\frac{4\log(T)}{\NG\fdpN[2]T_1}}\right) \leq O(T^{4/7} \log^{6/7}(T))
\]
}
\end{proof} 

\section{The multi-level policy}
\label{sec:llevel-details}

In this subsection, we analyze our $L$-level policy for $L > 3$, proving Theorems~\ref{thm:llevel-1} and Theorem~\ref{thm:llevel-2}. We first analyze it for the case of $K=2$ arms. The bulk of the analysis, joint for both theorems, is presented in \Cref{sec:Llevel-analysis-joint}. We provide two different endings where the details differ: \Cref{sec:Llevel-analysis-1} and
\Cref{sec:Llevel-analysis-2}, respectively. We extend the analysis to $K>2$ arms in \Cref{sec:Llevel-analysis-K}.

\OMIT{ 
\begin{theorem}
\label{thm:llevel}
The $L$-level recommendation policy gets regret
\[ O\left(T^{2^{L-1}/(2^L-1)} \log^2(T) \right)
\quad\text{for}\quad
L \leq L_{\max} :=  \log\rbr{ \frac{\ln T}{\log \NG^4} }.\]
In particular, if we pick $L = L_{\max}$, the regret is $O(T^{1/2} \polylog(T))$.
\end{theorem}
} 

\subsection{Joint analysis for $K=2$ arms (for both theorems)}
\label{sec:Llevel-analysis-joint}

We rely on the property (which holds in both theorems) that parameters $T_\ell$ satisfy
\begin{align}\label{eq:Llevel-params-assn}
 T_1 \leq \NG^4 \leq T_{\ell}\,/\,T_{\ell-1}
 \quad\text{for}\quad
  \ell\in \{ 2 \LDOTS L-1\},
 \end{align}
Wlog we assume $\mu_1 \geq \mu_2$ as the recommendation policy is symmetric to both arms.

Similarly to the proof of Theorem \ref{thm:3level}, we characterize some ``clean events".

\xhdr{Concentration of the number of arm $a$ pulls in the first level.}
For $a \in \{1,2\}$, define $\fdpN$ to be the expected number of arm $a$ pulls in one run of \ALGG used in the first level. By Lemma \ref{lem:greedy}, we know $\GdP \leq \fdpN \leq \GdT$ For group $G_{1,u,v}$, define $W_1^{a,u,v}$ to be the event that the number of arm $a$ pulls in this group is between $\fdpN T_1- \GdT \sqrt{T_1\log(T)}$ and $\fdpN T_1 + \GdT \sqrt{T_1\log(T)}$. By Chernoff bound,
\[
\Pr\sbr{W_1^{a,u,v}} \geq 1-2\,e^{-2\log T} \geq 1-2/T^2.
\]
Define $W_1$ to be the intersection of all these events (i.e. $W_1 = \bigcap_{a,u,v}W_1^{a,u,v}$). By union bound, we have
    $\Pr[W_1] \geq 1- 4\NG^2/T^2$.

\xhdr{Concentration of empirical mean for arm $a$ in the history observed by agent $t$.}
For each agent $t$ and arm $a$, imagine there is a tape of enough arm $a$ pulls sampled before the recommendation policy starts and these samples are revealed one by one whenever agents in agent $t$'s observed history pull arm $a$.  Define $W_2^{t,a,\z_1,\z_2}$ to be the event that the mean of $\z_1$-th to $\z_2$-th pulls in the tape is at most $\sqrt{\frac{3\log(T)}{\z_2-\z_1+1}}$ away from $\mu_a$. By Chernoff bound,
$\Pr\sbr{W_2^{t,a,\z_1,\z_2}} \geq 1 - 2e^{-6\log T} \geq 1- 2/T^6$.

Define $W_2$ to be the intersection of all these events (i.e. $W_2 = \bigcap_{t,a,\z_1,\z_2} W_2^{t,a,\z_1,\z_2}$). By union bound, we have
    $\Pr[W_2] \geq 1- 4/T^3$.

\xhdr{Anti-concentration of empirical mean of arm $a$ pulls in the $\ell$-th level  for $\ell \geq 2$.}
For $2\leq \ell \leq L-1$, $u\in [\NG]$ and each arm $a$, define $n^{\ell,u,a}$ to be the number of arm $a$ pulls in groups $G_{\ell,u,1},...,G_{\ell,u,\NG}$. Define $W_3^{\ell,u,a,high}$ as the event that $n^{\ell,u,a} \geq T_{\ell}$ implies the empirical mean of arm $a$ pulls in group $G_{\ell,u,1},...,G_{\ell,u,\NG}$ is at least $\mu_a + 1/\sqrt{n^{\ell,u,a}}$. Define $W_3^{\ell,u,a,low}$ as the event that $n^{\ell,u,a} \geq T_{\ell}$ implies the empirical mean of arm $a$ pulls in group $G_{\ell,u,1},...,G_{\ell,u,\NG}$ is at most $\mu_a - 1/\sqrt{n^{\ell,u,a}}$.

Define $H_{\ell}$ to be random variable the history of all agents in the first $\ell-1$ levels and which agents are chosen in the $\ell$-th level. Let $h_{\ell}$ be some realization of $H_{\ell}$. Notice that once we fix $H_{\ell}$, $n^{\ell,u,a}$ is also fixed.

Now consider $h_{\ell}$ to be any possible realized value of $H_{\ell}$. If fixing $H_{\ell}= h_{\ell}$ makes $n^{\ell,u,a}<T_{\ell}$, then $\Pr[W_3^{\ell,u,a,high} |H_{\ell} = h_{\ell}]=1$  If fixing $H_{\ell} = h_{\ell}$ makes $n^{\ell,u,a} \geq T_{\ell}$, by Berry-Esseen Theorem (\Cref{thm:app-concentration}(b)) and the fact that $\mu_a \in [1/3,2/3]$, we have
\[
\Pr\sbr{ W_3^{\ell,u,a,high} \mid H_{\ell}
    = h_{\ell}} \geq (1-\Phi(\nicefrac12)) - 5/\sqrt{T_{\ell}} > \nicefrac14.
\]
Similarly we also have
$\Pr\sbr{W_3^{\ell,u,a,low}\mid H_{\ell} = h_{\ell}}  > \nicefrac14$.

Since $W_3^{\ell,u,a,high}$ is independent with $W_3^{\ell,u,3-a,low}$ when fixing $H_{\ell}$, we have
\[
\Pr\sbr{ W_3^{\ell,u,a,high} \cap W_3^{\ell,u,3-a,low} \mid H_{\ell} = h_{\ell}}  > (\nicefrac14)^2 = 1/16.
\]
Now define $W_3^{\ell,a} = \bigcup_u (W_3^{\ell,u,a,high} \cap W_3^{\ell,u,3-a,low})$. Since  $(W_3^{\ell,u,a,high} \cap W_3^{\ell,u,3-a,low})$ are independent across different $u$'s when fixing $H_{\ell}=h_{\ell}$, we have
\[
\Pr\sbr{ W_3^{\ell,a}\mid H_{\ell}= h_{\ell}} \geq 1- (1-1/16)^\NG \geq 1 - 1/T^2.
\]
Since this holds for all $h_{\ell}$'s, we have $\Pr[W_3^{\ell,a}] \geq 1-1/T^2$. Finally define $W_3 = \bigcap_{\ell,a} W_3^{\ell,a}$. By union bound, we have
    $\Pr[W_3] \geq 1 - 2L/T^2$.

\xhdr{Anti-concentration of the empirical mean of arm $a$ pulls in the first level.}
For first-level groups $G_{1,u,1},...,G_{1,u,\NG}$ and arm $a$, imagine there is a tape of enough arm $a$ pulls sampled before the recommendation policy starts and these samples are revealed one by one whenever agents in these groups pull arm $a$. Define $W_4^{u,a,high}$  to be the event that first $\fdpN T_1 \NG$ pulls of arm $a$ in the tape has empirical mean at least $\mu_a + 1/\sqrt{\fdpN T_1 \NG}$ and define $W_4^{u,a,low}$  to be the event that first $\fdpN T_1\NG$ pulls of arm $a$ in the tape has empirical mean at most $\mu_a - 1/\sqrt{\fdpN T_1\NG }$. By Berry-Esseen Theorem (\Cref{thm:app-concentration}(b)) and the fact that $\mu_a \in [1/3,2/3]$, we have
\[
\Pr\sbr{W_4^{u,a,high}} \geq (1-\Phi(\nicefrac12)) - 5/\sqrt{\fdpN T_1\NG} > \nicefrac14.
\]
The last inequality holds if $T$ exceeds some constant.
Similarly,
    $\Pr\sbr{W_4^{u,a,low}} > \nicefrac14$.

Since $W_4^{u,a,high}$ is independent with $W_4^{u,3-a,low}$, we have
\[
\Pr\sbr{W_4^{u,a,high} \cap W_4^{u,3-a,low}}
    = \Pr\sbr{W_4^{u,a,high}} \cdot  \Pr\sbr{ W_4^{u,3-a,low}}>(1/4)^2 = 1/16.
\]
Now define $W^{a}_4$ as $\bigcup_u (W_4^{u,a,high} \cap W_4^{u,3-a,low})$. Notice that $(W_4^{u,a,high} \cap W_4^{u,3-a,low})$ are independent across different $u$'s. So, we have
    $\Pr[W^{a}_4] \geq 1- (1-1/16)^\NG \geq 1 -1/T^2$.
Finally we define $W_4 := \bigcap_{a} W^{a}_4$. By the union bound,
    $\Pr[W_4] \geq 1- 2/T^2$.

\medskip

Thus, we've defined four ``clean events" $W_1 \LDOTS W_4$ such that their intersection
    $W = \bigcap_{i=1}^4 W_i$
has probability $1-O(1/T)$. Consequently, the event $\neg W$
contributes at most $O(1/T) \cdot T = O(1)$ to the regret. Henceforth we assume $W$ happens.


By event $W_1$, we know that in each first-level group, there are at least $\fdpN T_1- \GdT \sqrt{T_1\log(T)}$ pulls of arm $a$. We prove in the next claim that there are enough pulls of both arms in higher levels if $\mu_1-\mu_2$ is small enough. For notation convenience, we set $\eps_0 = 1$,
    $\eps_1 = \frac{1}{4\sqrt{\fdpN T_1\NG}} + \frac{1}{4\sqrt{\fdpN[3-a] T_1\NG}}$
and $\eps_{\ell} = 1/(4\sqrt{T_{\ell}\NG})$ for $\ell \geq 2$.

\begin{claim}
\label{clm:l2_explore}
For any arm $a$ and level $\ell\in [2,L]$, if $\mu_1 - \mu_2 \leq \eps_{\ell-1}$, then for any $u \in [\NG]$, there are at least $T_{\ell}$ pulls of arm $a$ in groups
    $G_{\ell,u,v}$, $v\in [\NG]$,
and there are at least $T_{\ell}\cdot \NG(\NG-1)$ pulls of arm $a$ in the $\ell$-th level $\Gamma$-groups.
\end{claim}

\begin{proof}
We are going to show that for each level $\ell$ and arm $a$ there exists $u_a$ such that agents in groups
    $G_{\ell,\,u,\,u_a}$
and
    $\Gamma_{\ell,\,u,\,u_a}$, $u\in[\NG]$
all pull arm $a$. This suffices to prove the claim.

We prove the above via induction on the level $\ell$. 
We start by the base case when $\ell=2$. For each arm $a$, event $W_4$ implies there exists $u_a$ such that $W^{u_a,a,high}_4$ and $W^{u_a,3-a,low}_4$ happen. Fix some agent $t$ in groups
    $G_{2,\,u,\,u_a}$
and
    $\Gamma_{2,\,u,\,u_a}$, $u\in[\NG]$.
Events $W_4^{u_a,a,high}$,  $W_1^{a,u_a,v}$ and $W_2$ together imply that, letting
 $\Psi = \NG\cdot \sqrt{T_1\log(T)}$,
\begin{align*}
\bar{\mu}_a ^t
    &\geq \mu_a + \left(\fdpN T_1\NG \cdot \frac{1}{\sqrt{\fdpN T_1\NG}} - \GdT \,\Psi\cdot \sqrt{\frac{3\log(T)}{ \GdT \,\Psi}} \right)
    \cdot \frac{1}{\fdpN T_1\NG+ \GdT \,\Psi}
> \mu_a + \frac{1}{4\sqrt{\fdpN T_1\NG}}.
\end{align*}
The last inequality holds when $T$ is larger than some constant.
Similarly,
\[
\bar{\mu}_{3-a}^t< \mu_{3-a}   -
    \tfrac14\,/\,\sqrt{\fdpN[3-a] T_1\NG}.
\]
Denoting
    $\Lambda = \rbr{\fdpN T_1\NG}^{-1/2} + \rbr{\fdpN[3-a] T_1\NG}^{-1/2}$
for brevity, we have
\begin{align*}
\bar{\mu}^t_a - \bar{\mu}^t_{3-a} &> \mu_a - \mu_{3-a} +
    \Lambda/4
\geq -\eps_1+
    \Lambda/4
\geq
    \Lambda/8.
\end{align*}
By Assumption \ref{ass:embehave}, we have
\begin{align*}
\hat{\mu}_a^t - \hat{\mu}_{3-a}^t
> \bar{\mu}^t_a - \bar{\mu}^t_{3-a} -
       \estC\Lambda/\sqrt{2}
>
    \Lambda/8 - \estC\Lambda/\sqrt{2}
>0.
\end{align*}
The last inequality holds since $\estC$ is a small enough constant from Assumption \ref{ass:embehave}. Therefore, agents in groups
    $G_{2,\,u,\,u_a}$
and
    $\Gamma_{2,\,u,\,u_a}$, $u\in[\NG]$.
all pull arm $a$.

Now we consider the case when $\ell > 2$ and assume the claim is true for all smaller levels. For each arm $a$, event $W_3$ implies that there exists $u_a$ such that events $W^{\ell-1,u_a,a,high}_3$ and $W^{\ell-1,u_a,3-a,low}_3$ happen. Recall $n^{\ell-1,u_a,a}$ is the number of arm $a$ pulls in groups
    $G_{\ell-1,\, u_a,\, v}$, $v\in[\NG]$.
The induction hypothesis implies that $n^{\ell-1,u_a,a} \geq T_{\ell-1}$. Event $W^{\ell-1,u_a,a,high}_3$ together with the fact that $n^{\ell-1,u_a,a} \geq T_{\ell-1}$ implies that the empirical mean of arm $a$ pulls in groups
    $G_{\ell-1,\, u_a,\, v}$, $v\in[\NG]$.
is at least $\mu_a + 1/\sqrt{n^{\ell-1,u_a,a}}$. For any agent $t$ in groups
    $G_{\ell,\, u,\, u_a}$
and
     $\Gamma_{\ell,\, u,\, u_a}$, $u\in[\NG]$
it observes history of groups
    $G_{\ell-1,\, u_a,\, v}$, $v\in[\NG]$
and all groups in levels below $\ell-1$. Notice that in each group in the first $\ell-2$ levels, the number of agents is at most
\[ S := \NG^3\cdot (T_1 \GdT + T_2 + \cdots +T_{\ell-2})
    \leq T_{\ell-1}/(12\log(T)) \leq n^{\ell-1,u_a,a}/(12\log(T)).\]
By event $W_2$, we have
\begin{align*}
\bar{\mu}_a ^t &\geq \mu_a + \left(n^{\ell-1,u_a,a}  \cdot \frac{1}{\sqrt{n^{\ell-1,u_a,a} }}- S\cdot \sqrt{\frac{3\log(T)}{ S}} \right)
\cdot \frac{1}{n^{\ell-1,u_a,a}+ S}
> \mu_a + \frac{1}{4\sqrt{n^{\ell-1,u_a,a}  }}.
\end{align*}
The last inequality holds when $T$ larger than some constant.
Similarly, we prove
\[
\bar{\mu}_{3-a}^t < \mu_{3-a}   -\tfrac14 \,/\, \sqrt{n^{\ell-1,u_a,3-a}}.
\]
Denoting
    $\Lambda = \rbr{n^{\ell-1,u_a,a}}^{-1/2} +\rbr{n^{\ell-1,u_a,3-a}}^{-1/2}$
for brevity, we have
\begin{align*}
\bar{\mu}^t_a - \bar{\mu}^t_{3-a} &> \mu_a - \mu_{3-a}+
    \Lambda/4
\geq -\eps_{\ell-1}
    +\Lambda/4
\geq
   \Lambda/8.
\end{align*}
The last inequality holds because $n^{\ell-1,u_a,a}$ and $n^{\ell-1,u_a,3-a}$ are at most $T_{\ell-1} \NG$. By Assumption \ref{ass:embehave},
\begin{align*}
\hat{\mu}_a^t - \hat{\mu}_{3-a}^t
> \bar{\mu}^t_a - \bar{\mu}^t_{3-a}
        -\estC\Lambda
>  \Lambda/8-\estC\Lambda
>0.
\end{align*}
The last inequality holds since $\estC$ is a small enough constant from Assumption \ref{ass:embehave}.
So, agents in groups $G_{\ell,1,u_a},...,G_{\ell,\NG,u_a}$ and $\Gamma_{\ell,1,u_a},...,\Gamma_{\ell,\NG,u_a}$ all pull arm $a$.
\end{proof}

\begin{claim}
\label{clm:l2_exploit}
Fix any level $\ell\in [2,L]$
and assume $\eps_{\ell-1}\, \NG\leq \mu_1 - \mu_2 < \eps_{\ell-2}\, \NG$. Then arm $2$ is not pulled in groups with level $\ell,...,L$.
\end{claim}

\begin{proof}
We argue in 2 cases $\eps_{\ell-1} \sqrt{\NG} \leq \mu_1 - \mu_2 \leq \eps_{\ell-2}$ for $\ell \geq 2$ and $\eps_{\ell-2}  \leq \mu_1 - \mu_2 \leq \eps_{\ell-2} \sqrt{\NG}$ for $\ell > 2$. Since our recommendation policy's first level is slightly different from other levels, we need to argue case $\eps_{\ell-1} \sqrt{\NG} \leq \mu_1 - \mu_2 \leq \eps_{\ell-2}$ for $\ell=2$ and case $\eps_{\ell-2}  \leq \mu_1 - \mu_2 \leq \eps_{\ell-2} \sqrt{\NG}$ for $\ell =3$ separately. Thus, we have four cases to consider.

\begin{description}
\item[Case 1]  $\eps_{\ell-1} \NG \leq \mu_1 - \mu_2 \leq \eps_{\ell-2}$ for $\ell = 2$ (i.e. $\eps_1\NG\leq \mu_1 - \mu_2 \leq \eps_0$).

    We know agents in level at least 2 will observe at least $\fdpN T_1/2$ pulls of arm $a$ for $a \in \{1,2\}$. By event $W_2$, for any agent in level at least 2, we have
\[
|\bar{\mu}_a^t - \mu_a| \leq
    \sqrt{6\,\log(T)\,/\,(\NG\fdpN T_1)}.
\]
Denoting
    $\Lambda = \rbr{\NG\fdpN[1]T_1/2}^{-1/2} -  \rbr{\NG\fdpN[2]T_1/2}^{-1/2}$,
by Assumption \ref{ass:embehave} we have
\begin{align*}
\hat{\mu}_1^t - \hat{\mu}_2^t &\geq \bar{\mu}_1^t - \bar{\mu}_2^t
    -\estC\Lambda
\geq \mu_1 -\mu_2 - \Lambda\rbr{\sqrt{3\log T}+ \estC}
\geq \Lambda\rbr{ \tfrac{\NG}{4\sqrt{2}} -\sqrt{3\log T} - \estC }
>0.
\end{align*}
Therefore agents in level at least 2 will all pull arm 1.

\item[Case 2]  $\eps_{\ell-1} \NG \leq \mu_1 - \mu_2 \leq \eps_{\ell-2}$ for $\ell > 2$.

By claim \ref{clm:l2_explore}, for any agent $t$ in level at least $\ell$, that agent will observe at least $T_{\ell-1}$ arm $a$ pulls. By $W_2$, we have
\[
|\bar{\mu}_a^t - \mu_a| \leq \sqrt{3\log(T)\,/\,T_{\ell-1}}.
\]
By Assumption \ref{ass:embehave}, we have
\begin{align*}
\hat{\mu}_1^t - \hat{\mu}_2^t
    &\geq \bar{\mu}_1^t - \bar{\mu}_2^t - 2\estC/\sqrt{T_{\ell-1}}
    \geq \mu_1 -\mu_2 - 2 \sqrt{3\log(T)\,/\,T_{\ell-1}}- 2\estC/\sqrt{T_{\ell-1}} \\
    &\geq\sqrt{\frac{\NG}{16T_{\ell-1}}} -  2 \sqrt{\frac{3\log(T)}{T_{\ell-1}}}- \frac{2\estC}{\sqrt{T_{\ell-1}}}
    >0.
\end{align*}
Therefore agents in level at least $\ell$ will all pull arm 1.

\item[Case 3]  $\eps_{\ell-2} < \mu_1 - \mu_2 < \eps_{\ell-2}\NG$ for $\ell =3$ (i.e. $\eps_1 < \mu_1 - \mu_2 < \eps_1\NG$).

    By Claim \ref{clm:l2_explore}, for any agent $t$ in level at least $3$, that agent will observe at least $T_1\fdpN \NG^2/2$ arm $a$ pulls (just from the first level). By event $W_2$, we have
\[
|\bar{\mu}_a^t - \mu_a| \leq \sqrt{\frac{3\log(T)}{\NG^2\fdpN T_1/2}}.
\]
Denoting
    $\Lambda = \rbr{\NG^2\fdpN[1]T_1/2}^{-1/2}+\rbr{\NG^2\fdpN[2]T_1/2}^{-1/2}$
by Assumption \ref{ass:embehave} we have
\begin{align*}
\hat{\mu}_1^t - \hat{\mu}_2^t &\geq \bar{\mu}_1^t - \bar{\mu}_2^t
    -\Lambda\estC
\geq \mu_1 -\mu_2 - \Lambda\rbr{\sqrt{3\log T} - \estC}
\geq \Lambda\rbr{ \sqrt{\NG/2}/4- \sqrt{3\log T} - \estC}
>0.
\end{align*}
Therefore agents in level at least 3 will all pull arm 1.

\item[Case 4]  $\eps_{\ell-2} < \mu_1 - \mu_2 < \eps_{\ell-2}\NG$ for $\ell >3$.

Since $\mu_1-\mu_2 < \eps_{\ell-2}\NG < \eps_{\ell-3}$, by Claim \ref{clm:l2_explore}, for any agent $t$ in level at least $\ell$, that agent will observe at least $T_{\ell-2}\NG^2$ arm $a$ pulls (just from level $\ell-2$). By event $W_2$, we have
\[
|\bar{\mu}_a^t - \mu_a| \leq \sqrt{\frac{3\log(T)}{\NG^2T_{\ell-2}}}.
\]
By Assumption \ref{ass:embehave}, we have
\begin{align*}
\hat{\mu}_1^t - \hat{\mu}_2^t
    &\geq \bar{\mu}_1^t - \bar{\mu}_2^t - \frac{2\estC}{\sqrt{\NG^2T_{\ell-2}}}
    \geq \mu_1 -\mu_2 - 2 \sqrt{\frac{3\log(T)}{\NG^2T_{\ell-2}}}- \frac{2\estC}{\sqrt{\NG^2T_{\ell-2}}} \\
    &\geq\frac{1}{4\sqrt{\NG T_{\ell-2}}} -  2 \sqrt{\frac{3\log(T)}{T_{\ell-1}}}- \frac{2\estC}{\sqrt{T_{\ell-1}}}
    >0.
\end{align*}
Therefore agents in level at least $\ell$ will all pull arm 1.\qedhere
\end{description}
\end{proof}

\subsection{Finishing the proof of Theorem~\ref{thm:llevel-1} for $K=2$ arms}
\label{sec:Llevel-analysis-1}

We set the parameters $T_{\ell}$ for each level $\ell\in \{1 \LDOTS L-1\}$:
\begin{align}\label{eq:Llevel-params}
T_{\ell} = T^{\gamma_\ell}/\NG^3,
\quad\text{where}\quad
\gamma_\ell := \frac{2^{L-1} + 2^{L-2} + \cdots + 2^{L-\ell}}{2^{L-1}+ 2^{L-2} + \cdots + 1}
= \frac{2^L-2^{L-\ell}}{2^L-1}.
\end{align}

\noindent Note
    $T_{\ell} / T_{\ell-1} \geq T^{1/2^L} \geq \NG^4$
as required by \refeq{eq:Llevel-params-assn}.
Level $L$ has all remaining nodes:
\begin{align}\label{eq:Llevel-params-L}
T_L = (T-S)/\NG^3,
    \quad\text{where}\quad
    S:= T_1 \cdot \GdT\cdot \NG^2 - (T_2 + \cdots + T_{\ell-1}) \NG^3.
\end{align}


\noindent By Claim \ref{clm:l2_exploit}, the regret conditioned the intersection of clean events is at most
\begin{align*}
\max\cbr{T_1\GdT \NG^2 \;,\; \
    \max_{\ell \geq 2}
        \eps_{\ell-1}\cdot \NG\cdot S}
\leq  \max\cbr{ T_1\GdT \NG^2 \;,\; \max_{\ell \geq 2} 2 \eps_{\ell-1} T_{\ell} \NG^4 }
= O\left(T^{2^{L-1}/(2^L-1)} \log^2(T) \right).\qedhere
\end{align*}

\subsection{Finishing the proof of Theorem~\ref{thm:llevel-2} for $K=2$ arms}
\label{sec:Llevel-analysis-2}

We set the parameters as follows. The number of levels is $L=\log(T)/\log(\NG^4)$. For each level
$\ell\in \{1 \LDOTS L-1\}$ we have $T_{\ell} = \NG^{4\ell}$.
$T_L$ is defined via \refeq{eq:Llevel-params-L}. Note  these settings satisfy \refeq{eq:Llevel-params-assn}, as required.
\OMIT{ 
\begin{corollary}
\label{cor:llevel}
With the proper setting of $L$ and $T_1,...,T_L$ described above, the $L$-level recommendation policy gets regret $O(\min(1/\Delta, T^{1/2})\polylog(T))$. Here $\Delta = |\mu_1 -\mu_2|$ and the $L$-level recommendation policy does not need to know $\Delta$. Moreover, agent $t$ observes a subhistory of size at least $\Omega( \lfloor t/\polylog(T)\rfloor)$.
\end{corollary}

Notice that in the proof of Theorem \ref{thm:llevel}, by the end of Claim \ref{clm:l2_exploit}, the only constraint we need about $T_{\ell}$'s is that $T_{\ell} / T_{\ell-1} \geq \NG^4$ for $\ell=2,...,L-1$ and $T_1 \geq \NG^4$. And our new settings of $T_{\ell}$'s still satisfy this constraint. So we can reuse the proof of Theorem \ref{thm:llevel} till the end of Claim \ref{clm:l2_exploit}.
}
%
Recall from \Cref{sec:Llevel-analysis-joint} that $\eps_{\ell} =\Theta(1/\sqrt{T_{\ell} \NG})$ for $\ell \in [L-1]$ and $\eps_0 = 1$.
Then if $\Delta < \eps_{L-1} \NG$, notice that even always picking the sub-optimal arm gives expected regret at most $T(\mu_1-\mu_2) = T\Delta = O(T^{1/2} \polylog(T))$. On the other hand, $T^{1/2} = O(\polylog(T)/\Delta)$.
    So, regret is $O(\min(\nicefrac{1}{\Delta}, T^{1/2})\polylog(T))$.
Otherwise $\Delta \geq \eps_{L-1} \NG$. In this case, we can find $\ell \in \{2,...,L\}$ such that $\eps_{\ell-1} \NG\leq \Delta < \eps_{\ell-2} \NG$. By Claim \ref{clm:l2_exploit}, we can upper bound the regret by
\begin{align*}
&\Delta \cdot \rbr{T_1 \GdT \NG^2  +T_2 \NG^3+ \cdots T_{\ell-1} \NG^3}
=O\rbr{\NG^3\Delta T_{\ell-1}}
=O\rbr{\NG^7\Delta T_{\ell-2}} \\
&\qquad=O\rbr{\NG^6\Delta\cdot  \eps^{-2}_{\ell-2}}
=O\rbr{\NG^8\Delta \cdot \Delta^{-2}}
=O\rbr{\polylog(T)/\Delta}.
\end{align*}
We also have $1/\Delta \leq 1/(\eps_{L-1}\NG)= O(T^{1/2})$.
So, regret is $O(\min(\nicefrac{1}{\Delta}, T^{1/2})\;\polylog(T))$.

Finally to analyze subhistory sizes, note that agents in level $\ell$ observe the history of all agents at or below level $\ell-2$. Furthermore, the ratio between the number of agents below level $\ell$ and the number of agents below level $\ell-2$ is bounded by $O(\polylog(T))$, implying the result.


\subsection{Extending the analysis to $K>2$ arms.}
\label{sec:Llevel-analysis-K}
Here we discuss how to extend Theorems~\ref{thm:llevel-2} and~\ref{thm:llevel-2} to $K>2$ arms. The analysis is very similar to the $K=2$ case, so we only sketch the necessary changes.

\OMIT{ 
\begin{theorem}
\label{thm:constarm}
Theorem \ref{thm:llevel} and Corollary \ref{cor:llevel} can be extended to the case when $K$ is constant larger than 2. In the extension of Corollary \ref{cor:llevel}, $\Delta$ is defined as the difference between means of the best and the second best arm.
\end{theorem}
} 

We still wlog assume arm 1 has the highest mean (i.e. $\mu_1 \geq \mu_a, \forall a \in \A$. We first extend the clean events $W_1,W_2,W_3,W_4$ in \Cref{sec:Llevel-analysis-joint} to the case $K>2$. Events $W_1$ and $W_2$ extend naturally: we still set $W_1 = \bigcap_{a,s}W_1^{a,s}$ and $W_2 = \bigcap_{t,a,\z_1,\z_2} W_2^{t,a,\z_1,\z_2}$.
For event $W_3$, we change the definition $W_3^{\ell,a} = \bigcup_u \left(W_3^{\ell,u,a,high}  \cap \left(\bigcap_{a' \neq a} W_3^{\ell,u,a',low}\right) \right)$ and $W_3 = \bigcap_{\ell,a} W_3^{\ell,a}$. Event $W_4$ is extended similarly: $W^{a}_4 := \bigcup_u \left(W_4^{u,a,high} \cap \left(\bigcap_{a' \neq a} W_4^{u,a',low}\right) \right)$ and $W_4 = \bigcap_a W^a_4$. Since $K$ is a constant, the same proof technique shows that the intersection of these clean events happen with probability  $1-O(1/T)$. So the case when some clean event does not happen contributes $O(1)$ to the regret.

Now we proceed to extend Claim \ref{clm:l2_explore} and Claim \ref{clm:l2_exploit}. The statement of Claim \ref{clm:l2_explore} should be changed to ``For any arm $a$ and $2\leq \ell \leq L$, if $\mu_1 - \mu_a \leq \eps_{\ell-1}$, then for any $u \in [\NG]$, there are at least $T_{\ell}$ pulls of arm $a$ in groups $G_{\ell,u,1},G_{\ell,u,2}, ... ,G_{\ell,u,\NG}$ and there are at least $T_{\ell}\NG(\NG-1)$ pulls of arm $a$ in the $\ell$-th level $\Gamma$-groups''. The statement of Claim \ref{clm:l2_exploit} should be changed to ``For any $2 \leq \ell \leq L$, if $\eps_{\ell-1} \NG\leq \mu_1 - \mu_a < \eps_{\ell-2} \NG$, there are no pulls of arm $a$ in groups with level $\ell,...,L$.''

The proof of Claim \ref{clm:l2_exploit} can be easily changed to prove the new version by changing ``arm 2'' to ``arm $a$''. The proof of Claim \ref{clm:l2_explore} needs some additional argument. In the proof of Claim \ref{clm:l2_explore}, we show that $\hat{\mu}_a^t - \hat{\mu}_{3-a} > 0 $ for agent $t$ in the chosen groups. When extending to more than 2 arms, we need to show $\hat{\mu}_a^t - \hat{\mu}_{a'}^t > 0$ for all arm $a' \neq a$. The proof of Claim \ref{clm:l2_explore} goes through if $\mu_1- \mu_{a'} \leq \eps_{\ell-2}$ since then there will be enough arm $a'$ pulls in level $\ell-1$. We need some additional argument for the case when $\mu_1 - \mu_{a'} > \eps_{\ell-2}$. Since $\mu_1- \mu_{a'} > \eps_{\ell-2} > \eps_{\ell-1}\NG$, we can use the same proof of Claim \ref{clm:l2_exploit} (which rely on Claim \ref{clm:l2_explore} but for smaller $\ell$'s) to show that there are no arm $a'$ pulls in level $\ell$ and therefore $\hat{\mu}_a^t - \hat{\mu}_{a'}^t > 0$.

Finally we proceed to bound the regret conditioned on the intersection of clean events happens. The analysis for $K=2$
bounds it by consider the regret from pulling the suboptimal arm (i.e. arm 2). When extending to more than 2 arms, we can do the exactly same argument for all arms except arm 1. This will blow up the regret by a factor of $(K-1)$ which is a constant.

\newpage
\section{Proof of Lemma~\ref{lem:expts-canon}}
\label{app:expts-proof}


Consider a \FDP of length $P$. The ``reward tape" of a given arm $a$, denoted $\mT_a$, is an array of length $P$ such that each entry $\mT_{a,j}$,
$j\in [P]$ is the realized reward of arm $a$ when/if this arm is chosen for the $j$-th time. Say that $\mT_a$ \emph{dominates} another reward tape $\mT'_a$ if there is a weak inequality $\mT_{a,j}\geq \mT'_{a,j}$ for each entry $j$.

Fix the reward tapes for both arms and both problem instances such that $\mT_1(\mI')$ dominates $\mT_1(\mI)$ and $\mT_2(\mI)$ dominates $\mT_2(\mI')$. (In words, changing from $\mI$ to $\mI'$ rewards for arm $1$ weakly increase, and rewards for arm $2$ weakly decrease.) Let $M_{t,a}(\mI)$ be the (realized) number of times arm $a$ is chosen by time $t$ in problem instance $\mI$.

We claim $M_{t,1}(\mI')\geq M_{t,1}(\mI)$ for all rounds $t\in [P]$. This is proved by induction on $t$. Indeed, there's equality in round $1$. For the inductive step, assume $M_{t,1}(\mI')\geq M_{t,1}(\mI)$ in a given round $t$. Then either we have a strict inequality, in which case the weak inequality trivially holds for the next round $t+1$, or we have equality, in which case arm $1$ being chosen in instance $\mI$ implies it is chosen in $\mI'$, too. Claim proved.

To complete the proof, we observe that the reward tapes can be correlated across the two problem instances so that their realizations align as assumed above. (This is because $\mu_1$ weakly increases from $\mI$ to $\mI'$, and $\mu_2$ weakly decreases.)

\end{appendices}


\end{document}